\documentclass[11pt]{article}
\usepackage{enumerate,amsmath,graphics,amssymb,graphicx,
amscd,amscd,amsbsy,multirow,float,booktabs,verbatim,xy}

\usepackage[top=25mm, bottom=25mm, left=37.5mm, right=37.5mm]{geometry}

\usepackage{graphicx}
\graphicspath{{figs/}{figs2/}}
\usepackage{import}

\usepackage{amsmath,graphics,amssymb,graphicx,amscd,xypic,
amscd,amsbsy,multirow,float,
booktabs,verbatim,url}
\usepackage[round]{natbib}
\usepackage{enumerate}

\newcommand{\QED}{\hspace*{\fill}\rule{2.5mm}{2.5mm}}
\newtheorem{theorem}{Theorem}[section]
\newtheorem{definition}{Definition}[section]

\newenvironment{proof}{\noindent{\bf Proof\ }}{\QED\\}

\newtheorem{lemma}{Lemma}[section]

\newtheorem{corollary}{Corollary}[section]

\newcommand{\R}{\mathbb{R}}

\newcommand{\PL}{\mathcal{PL}}
\newcommand{\PPL}{\mathcal{PPL}}

\newcommand{\LB}{\mathcal{LB}}
\newcommand{\PLB}{\mathcal{PLB}}
\newcommand{\ALB}{\mathcal{ALB}}
\newcommand{\PALB}{\mathcal{PALB}}

\newcommand{\F}{\mathcal{F}}

\begin{document}

\title{A framework for fitting sparse data}
\author{ Reza Hosseini$^{1}$, Akimichi Takemura$^2$, Kiros Berhane$^3$\\
IBM Research$^{1}$, University of Tokyo$^2$, University of Southern California$^3$\\
$^1$rezah@sg.ibm.com
 }

\maketitle

%\newpage

%\newpage
\begin{abstract}

This paper develops a framework for fitting functions with domains in the Euclidean space, 
when data are sparse but a slow variation allows for a useful fit. We measure the
variation by Lipschitz Bound (LB) -- functions which admit smaller LB are considered to vary more slowly. Since
most functions in practice are wiggly and do not admit a small LB, we extend this framework by approximating a
wiggly function, $f$,  by ones which admit a smaller LB and do not deviate from $f$ by more than a specified
Bound Deviation (BD). In fact for any positive LB, one can find such a BD,
thus defining a trade-off function (LB-BD function) between the variation measure (LB) and the deviation measure (BD).
We show that the LB-BD function satisfies nice properties: it is non-increasing and convex. We also present a method to obtain it using convex optimization. For a function with given LB and BD, we find the optimal
fit and present deterministic bounds  for the prediction error of various methods. 
Given the LB-BD function, we discuss picking an
appropriate LB-BD pair for fitting and calculating the prediction errors.  The developed methods can naturally
accommodate an extra assumption of periodicity to obtain better prediction errors. Finally we present the
application of this framework to air pollution data with sparse observations over time.
\end{abstract}

\vspace*{.3in}


\section{Introduction}
\label{sect:introduction}

This paper investigates the problem of approximating (fitting) functions when data are sparse over time
or spatial domains of data.
Such {\it data-sparse} situations are often encountered when collecting large amounts of data is
expensive or practically implausible. For example in many air pollution studies including the {\it Southern
California Children Health Study}, (  \cite{paper-exposure-franklin} and
  \cite{paper-exposure-gauderman-2007}), only sparse data are collected over time for concentrations of Ozone in some homes and schools in Southern California to assess the effect
of air pollution exposure on children's lung function. Using such sparse data, we are interested in approximating
the exposure for a given location over a time period of interest.  In such cases
some properties of the data might allow for a good approximation (prediction/fit) despite the sparse data structure. For example for Ozone concentrations in Southern California, bi-weekly measurements (the measuring filters are installed and
collected in such periods) are available and therefore the process over time varies slowly. Left panel of Figure \ref{O3_biweekly_UPL_2004_2007.pdf} depicts the biweekly moving average of Ozone concentration for a central site (Upland, CA)  where complete data are available during 2004--2007.
Other properties of the process might also help us fit the function in sparse data situations. For example many processes over time show an approximate periodic pattern on annual scale  (e.g.\;weather and air pollution). The approximate periodicity for the Ozone process in Southern California can also be seen in left panel of Figure \ref{O3_biweekly_UPL_2004_2007.pdf}. This work utilizes such properties to  improve the methods of fitting.

When we are working with (at least one-time) differentiable real functions defined on the $d$-dimensional Euclidean space $\R^d$, we can naturally define a measure of
variation of the function $f$  by the supremum of its first-order derivative (or gradient for multidimensional case) on the
domain: $\sup ||f'(x)||,\;x \in D$, where $||.||$ is the Euclidean norm. Of course this definition is not useful for most processes we encounter in the real world -- even if they show some global slow-variation -- because often there are irregular small variations which make the function non-differentiable (left panel of Figure \ref{O3_biweekly_UPL_2004_2007.pdf}).

The key concept we use
in this paper is a measure  of variation (or roughness) for general non-differentiable functions on a given
domain. At first we consider functions which admit a Lipschitz Bound (LB) on the specified domain and assign the
variation of the function to be the infimum of all such bounds. A function $f:D\subset \R^d \rightarrow \R$,
where $D$ is a subset of $\R^d$ is said to have Lipschitz Bound (LB), $m,$ if
$|f(x)-f(y)|\leq m ||x-y||,$ where $||.||$ denotes $L^2$ norm. The interpolation of functions with a given
Lipschitz Bound is also considered in \cite{Gaffney-1976}, \cite{Sukharev-1978}, \cite{Beliakov-2006}, \cite{Sergeyev-2010} and in these works the
 {\it optimal central algorithm} was developed which minimizes the prediction error. \cite{Beliakov-2006} developed a fast algorithm for computing central optimal interpolant.  The Lipschitz framework immediately includes piece-wise differentiable functions but this generalization is still not
adequate (useful) for processes we encounter in practice. This is because such processes do not admit a small enough Lipschitz Bound for the fits or the prediction errors to be useful. We call such functions, ``wiggly'' functions. (Note that this is not an accurate
mathematical definition.) As one of the contributions of this work, we extend this
framework by approximating a wiggly function $f:D\subset \R^d \rightarrow \R,$ which does not admit a small enough LB by another function $g$, which does admit a small LB and deviates from $f$ only by a small {\it Bound Deviation} (BD), $\sigma$ in terms of the sup norm: $||f-g||_{\infty}=\sup_{x \in D}|f(x)-g(x)|.$ Then we find the optimal approximation of a given function $f$ with known LB and BD and provide the prediction (approximation) errors for the optimal solution and other standard approximation methods, thus extending the results in \cite{Gaffney-1976}, \cite{Sukharev-1978}, \cite{Beliakov-2006},  \cite{Sergeyev-2010}  to a much more practical class of functions.

Another key observation is: for each given LB=$m$, (which may not be satisfied by $f$), we can consider all the functions $g$
which satisfy the LB, $m$, and calculate the infimum of the distance of all those function from $f$ (in
terms of the supremum norm) and denote it by $\gamma_f(m)$.  Thus we can construct a generalized concept of LB in which
a given function can be considered to have any $m\geq 0$ as  LB, albeit up to a Bound Deviation, $\gamma_f(m)$, which is the price one pays for getting $m$ as LB. We call this trade-off function, $\gamma_f$,
the LB-BD function (or curve) of $f$.  This concept is similar but not identical to the {\it bias--variance} trade-off considered in statistical learning considered in \;  \cite{book-hastie-tib}. Here we thoroughly study the properties of the LB-BD trade-off function and explicitly use it in picking appropriate LB and BD for prediction. This methodology may be also useful for the methods which use bias--variance trade-off -- a comprehensive study of which is outside the scope of this paper.  For the simulation and applications, we mainly focus on the 1-dimensional (1-d) domain case. This framework  can also be considered for the multidimensional case, which deserves an extensive analysis that is beyond the scope of this paper.

\begin{figure}
\centering
\includegraphics[width=0.4\textwidth]{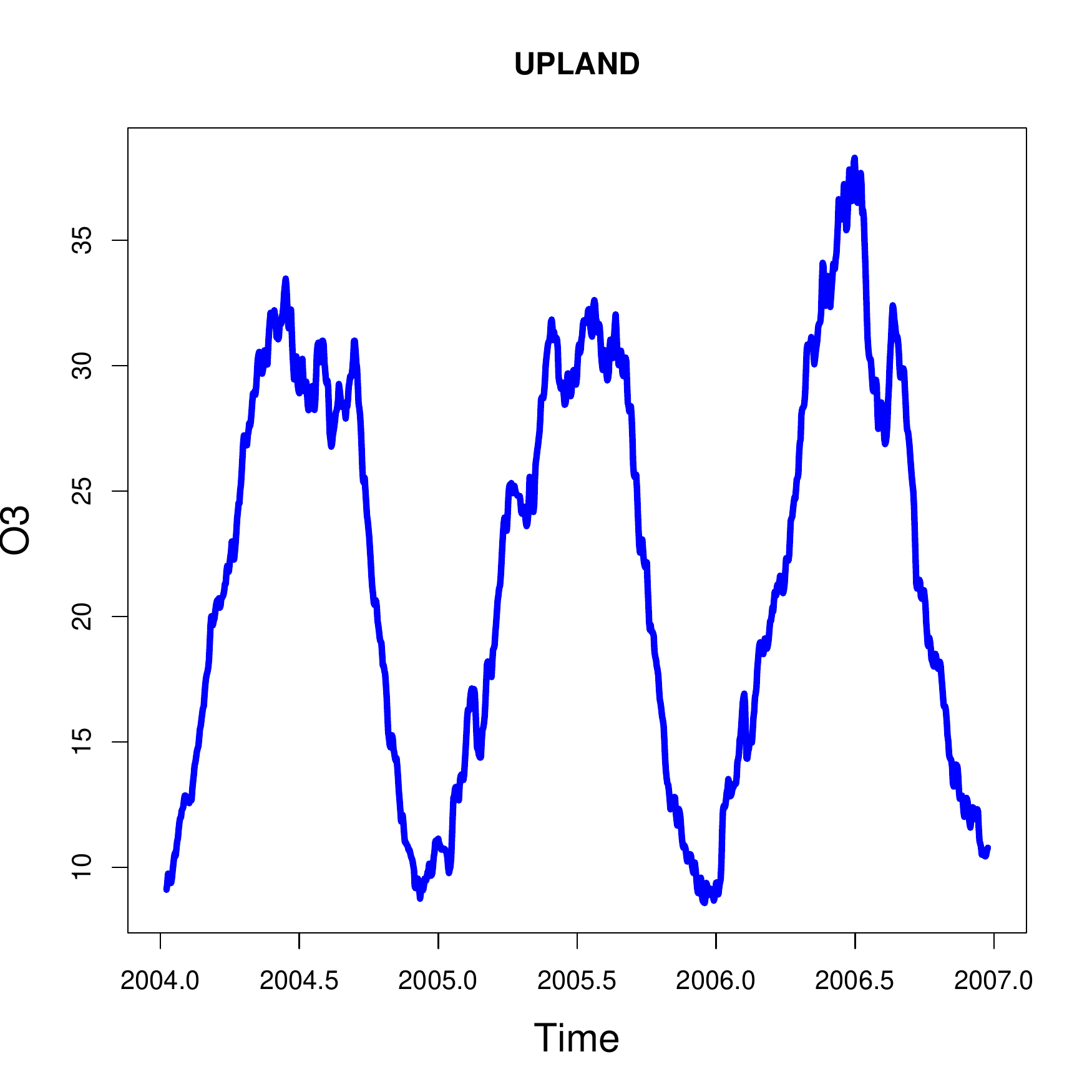}\includegraphics[width=0.4\textwidth]{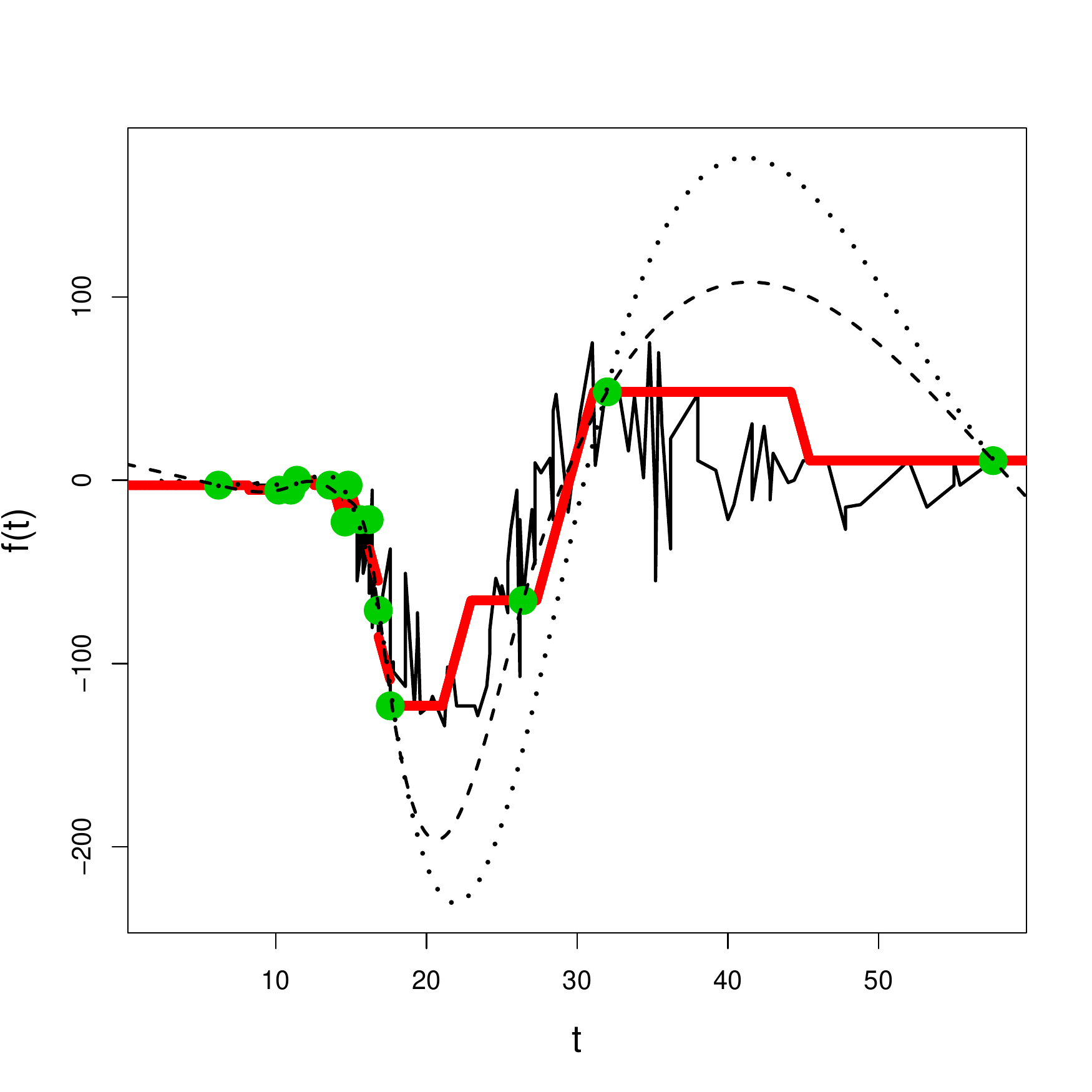}

\caption{\small (Left Panel) Biweekly moving average of Ozone concentrations ($O_3$) plotted  at a central
station in Upland(UPL) in Southern California. We observe a general pattern which varies slowly over time and some small variations on top of that.
The process is also approximately periodic at least when we focus on one given year (i.e.\;the beginning and end values approximately are the same). (Right Panel) Motorcycle head acceleration data. The data points (filled circles) are sampled from the full data
 (thin black curve).   The LOESS (dotted) and the smoothing spline (dashed)
 fits are poor for values bigger than 30 (where the data are sparse) and goes way beyond the range of data points,
 while $Lipfit$ method introduced in this work (thick curve) still achieves a reasonable fit.} 

\label{O3_biweekly_UPL_2004_2007.pdf}
\label{Lipfit_method_with_err_mcycle_example.pdf}
\end{figure}

\begin{comment}
In order to assess the performance of approximation methods, we need to define appropriate error measures. The
definition of the error should depend on the specific application. Two important cases are as follows: (1)
point-wise approximation; (2) integral approximation. In (1), the goal is to achieve good approximations to the
function $f$ at all points in the domain, while in (2) we are interested in approximating the integral of the function $f$ on
the domain $D$. Therefore different errors should be considered accordingly. For example we can use
$\sup_{x\in D}|f(x)-\hat{f}(x)|$ for the first case and $|\int_D f(x)dx - \int_D \hat{f}(x)dx|$ for
the second case. In (2) the distinction between interpolation and approximation becomes less important and in
fact, simple averaging of the available data ($AVG$), can be considered a reasonable method for that purpose and is often used in practice for example in
the study of air pollution exposure assessment in \cite{paper-exposure-franklin}. However even for the integral approximation case, we developed methods
which outperform the simple averaging (and any other possible method) -- a fact we show both by theory and simulations.
\end{comment}

Since we assume the data are very sparse in some subregions of the domain, the use of classical statistical methods such as regression may not be suitable. This is because with only few data points, it is either impossible to estimate the trends and the
error (due to having too many parameters) or the estimates will be extremely poor, resulting in issues such as
``over-fitting''.  As an example consider the data consisting of a series of measurements of head accelerations ($y$-axis)
versus $time$ ($x$-axis) (Figure \ref{Lipfit_method_with_err_mcycle_example.pdf}, right panel) in a simulated motorcycle
accident, used to test crash helmets (see \cite{paper-smoothing-silverman}). These data are available as a part
of the R package library, {\it MASS}. The full data set is depicted (black curve) and we have chosen a subset of
15 points (filled circles). The {\it locally weighted regression} (LOESS) fit (dotted), (e.g.\; \cite{book-loess-cleveland}), and {\it smoothing spline} fit (dashed), (e.g.\; \cite{book-hastie-tib}) are also given. We observe that while these methods perform well at the beginning of the series, where more data are available,
they fail dramatically at larger values (larger than 30), where less data are available. 
%%The LOESS and smoothing spline fits are fitted using R packages which estimate the parameters automatically with the standard available techniques, from which
%%the predicted curves are created. 
In contrast, the thick curve is created using the {\it Lipfit} method
developed in this paper, performs better by tracking the data closely. The problem with most of the classical
curve fitting methods (regression/regularization) -- when applied to the data which are sparse in some subregions -- is that
there is nothing to prevent their fits from going well beyond the range of the data as shown in the above
example. On the other hand $Lipfit$ is guaranteed to stay within the data range by definition. In general methods which produce fits which stay within the range of the data are desirable and can
be useful in data sparse situations. We say a method is {\it data-range faithful} if it does not go beyond the available
data range. 

The remainder of the paper is organized as follows. Section \ref{sect:background} includes a historical background of data fitting and interpolation methods
and outlines their connection with this work.
Section \ref{sect:theory} develops a framework for approximating
(fitting) slow-moving curves which are defined using the Lipschitz Bound and Bound Deviation. We also propose
several loss functions to assess the goodness of approximation methods. We  discuss various approximation methods (some of which are interpolation methods) and find an optimal one: $Lipfit$.
%We find optimal sampling points for getting the best possible approximation
%error.  
Section \ref{sect:simulation} compares the discussed approximation methods using simulations.  Section  \ref{sect:LB-BD-trade-off} discusses the trade-off between the LB and BD for a
given function; defines the LB-BD function associated with a given function; and shows the nice properties of the LB-BD curve such as convexity. It also discusses methods for calculating the LB-BD function, including a convex optimization method. We also discuss 
finding appropriate parameters (LB-BD) for applying the methods and calculating approximation errors in practice, including a {\it Prediction Error Minimization Method} (PEM). Section
\ref{sect:application} describes the application of the method in measuring Ozone exposure. Finally Section
\ref{sect:discussion} discusses some remaining issues and extensions.

\section{Background}
 \label {sect:background}

Throughout this paper we use the words function approximation, fitting and prediction, interchangeably and to
refer to any method which inputs data and outputs values for the function at unknown points. However it
is useful to clarify what is usually meant by interpolation here and in the relevant literature because
several of the methods discussed in this work are interpolation
methods. Interpolation refers to any method which inputs $n$ values of a {\it target function} $f$ defined
on $D\subset \R^d$: $(x_1,f(x_1)),\cdots,(x_n,f(x_n))$ and outputs a function $\hat{f}$ on $D$ which {\it agrees}
with $f$ on the given points (and it is supposed to be close to $f$ values outside the given points). This is in
contrast to classical fitting methods in statistics (e.g.\;linear regression) for which often the estimated curve
does not go through the given points. In order to include all possible methods, we use the term {\it
approximation} to refer to any method that given the input data returns a function $\hat{f}$ on $D$.
A very simple example of approximation -- which is not an
interpolation method -- is a method we denote by $AVG$ and simply takes the average of the available values of $f$
and assign that to all the domain: $\hat{f}(x)=\sum_{i=1}^n f(x_i)/n$. We do not think the distinction between
interpolation and general approximation is really useful since any approximation method can be slightly tweaked
to become an interpolation method by redefining the value of the approximation at the available points to be the
same as the data. Moreover there are usually infinitely many {\it out-of-sample} points as compared to finitely many
{\it in-sample} points. Therefore in practice, we often care mostly about the out-of-sample performance anyway.

Since we consider some interpolation methods in this paper, here we discuss some historical background on this topic.
Interpolation of points surprisingly goes back to astronomy  in ancient Babylon and Greece when it was all about
time keeping and predicting astronomical events (  \cite{paper-interpolation-meijering}). Later Newton and Lagrange
studied the problem of interpolating a function $f$ defined on an interval $[a,b]$ with given values on $n$ points:
$x_1,\cdots,x_n$, by a polynomial of degree $n$ and arrived at the same solution (with different computational
methods). The Lagrange method gives this polynomial, $p$, by defining $l_i(x)=\prod_{j=1;j\neq i}^n
\frac{x-x_i}{x_j-x_i}$ and letting $p_n(x)=\sum_{i=1}^n l_i(x)f(x_i).$
Moreover it can be shown (see \cite{book-numerical-cheney}) that if $f$ is $(n+1)$ times differentiable and $|f^{(n+1)}|\leq M$, then\begin{equation}
|f(x)-p_n(x)|\leq \frac{1}{(n+1)!}M\prod_{i=1}^n |x-x_i|.\label{eqn-err-bound-newton}
\end{equation} 
Unfortunately both the
existence of $f^{(n+1)}$ and having a small bound are rare in practice. While a practitioner may have an idea about the first derivative magnitude (or other measures of variation such as Lipschitz Bound)
of the process under  study, it is extremely
rare to know something about the $(n+1)$th derivative of the process or even believe it exists! To
illustrate this point consider the bounded function $f(x)=(1+x^2)^{-1}$ for which derivatives of
all orders are available and suppose $n$ equally-spaced points are available for interpolation. One can show
\[\lim_{n \rightarrow \infty} \max_{x \in [-5,5]} |f(x)-p_n(x)| = \infty,\]
(  \cite{book-numerical-cheney}). This means as more data become available, the accuracy of the interpolation
gets worse! The reason the Newton/Lagrange polynomial method fails dramatically in this case is the high-order
derivatives of this simple bounded function become very large for some $x \in [-5,5]$ (or else Equation
\ref{eqn-err-bound-newton} would guarantee a precise bound). Apparently it was expected that a (continuous)
function $f$ will be well-approximated by interpolating polynomials and ``in the history of numerical
mathematics, a severe shock occurred when it was realized that this expectation was ill-founded''
(  \cite{book-numerical-cheney}). Based on the discussion above, we seek methods which make the least possible
assumptions regarding the properties of function $f$. In fact we do not require existence of any derivatives and
only require the function to have a Lipschitz Bound, which is a weaker assumption than that of the
existence of the first derivative. Also we derive the approximation errors merely based on this bound. Later we even relax the existence of a (small) Lipschitz Bound by allowing the function to be well-approximated by
a function with relatively small Lipschitz Bound up to a deviation defined appropriately.

Another view of approximating functions emerged with the least squares method of Gauss which was followed by various other regression methods. The main difference of
this framework (from the interpolation methods previously developed) is to view the function as a combination of a true underlying function and some added noise. For example a version of linear regression assumes $f(x)=\beta_0 + \beta_1 x + \epsilon,$ where $\epsilon$ is independent and identically distributed noise process, for example normally distributed: $\epsilon  \sim N(0,\sigma^2)$.  This can be generalized in many ways to include more predictors or non-linear trends. For example $f(x)=\beta_0 + \beta_1 B_1(x) + \beta_2 B_2(x) + \epsilon,$
for some basis functions  $B_1(x), B_2(x)$. The function can now be seen as an imperfect observation
of a true function and the main objective is to infer about the true function. For example if we predict a value at an observed $x=x_i$ for which $f$ is observed to be $f(x_i)$, using the historical interpolation methods, we get back the same observed value $f(x_i)$, while from the regression, we may get back a completely different value, which is supposed to be closer to the noise-free true value.

A more recent view of fitting functions emerged as the so-called regularization methods (  \cite{book-hastie-tib}). As an example, consider the {\it smoothing splines} method which finds the minimizer of
\begin{equation}
\underset{f}{argmin} \sum_{i=1}^n |f(x_i)-y_i|^2 + \lambda \int_{D} ||f''(x)||dx,
\label{eqn-smoothing-splines}
\end{equation}
where $f''(x)$ is a second derivative and $\lambda\geq0$ is a {\it penalty} term, which creates a trade-off between the deviation  of the fitting (approximation) function $f$ and
the variation of the function. If we do not include the second term: $\lambda \int_{D} ||f''(x)||dx$, we end up with a function that necessarily interpolates them and can have chaotic behavior outside the observed points -- a similar problem to that of interpolation methods of Lagrange and Newton. An appropriate $\lambda$ is usually chosen by cross-validation (e.g.\; \cite{book-hastie-tib}). Solving Equation \ref{eqn-smoothing-splines} can be shown to be equivalent to
\begin{eqnarray}
\underset{f}{argmin} \sum_{i=1}^n |f(x_i)-y_i|^2,\;\;\;\;\;\;\; \mbox {subject to}\;\; \int_D ||f''(x)||dx \leq \lambda^{\star},
\label{eqn-smoothing-splines2}
\end{eqnarray}
for some $\lambda^{\star}\geq0,$ which is determined by $\lambda$. In fact the second representation comes from
the dual problem of the first one in convex optimization theory. Conceptually the regularization methods, such as smoothing splines, do not explicitly assume and model a noise
process but rather prevent the function to vary too much through $\int_D
||f''(x)||dx \leq \lambda^{\star}$.  The framework we use in this paper is similar in this sense and does that by
controlling the variation through the Lipschitz Bound and the deviation by  $\max_{i=1}^n |f(x_i)-y_i|$.  Of
course the solution in the two cases can be dramatically different. One thing the latter achieves is the fit
always remains within the range of the data through controlling the Lipschitz Bound and  the severe deviation
measure of $\max_{i=1}^n |f(x_i)-y_i|,$ which does not let the approximation to deviate from the data much at any
individual point. 

\begin{comment}
With this background comparison it becomes clear that other solutions to the fitting problem
can be considered by choosing various variation measures and deviation measures -- each of which may suit
different applications. A thorough study of these various choices is beyond the scope of this paper and we discuss more about
this idea in a  general framework in the discussion section.
\end{comment}

\section{A framework for approximating slow-moving functions}
\label{sect:theory}
Making inference about an arbitrary function with sparse data is not feasible if we do not have any extra
information at our disposal. Here we show that one such useful assumption is having a relatively small {\it Lipschitz Bound} which is defined formally below. Here we consider real functions defined on a subset $D$ of the $d$-dimensional Euclidean space, $\R^d$: $f: D \subset \R^d \rightarrow \R,$
and denote the set of all such functions by $\R^{D}.$ We also consider the supremum norm on this space $\|f\|_{\infty}={\sup}_{x\in D}{|f(x)|},$ which induces a metric (and topology) on $\R^{D}$. We denote the Euclidean norm for a vector $v$ in $\R^d$ by $||v||.$ 
\begin{definition} Suppose $f: D \subset \R^d \rightarrow \R$ is a function.\\
(i) $f$ is said to have a Lipschitz Bound (LB), $m$, if
$|f(x)-f(y)|\leq m||x-y||,\;x,y \in D$. We denote the set of all such functions by $\LB(D,m)$
or $\LB(m)$ when the domain is clear from the context.\\
(ii) For $d=1$, we denote the set of all {\it periodic} functions on $[a,b]$ ($f(a)=f(b)$) and with Lipschitz
Bound $m$
by $\PLB([a,b],m)$ or $\PLB(m),$ if the domain is  clear from the context. \\
(iii) Infimum Lipschitz Bound:
\[Lip(f)=\inf\{m \in \R,\;\;  |f(x)-f(y)|\leq m ||x-y||,\;\;\forall x,y \in D\}.\]
\end{definition}

Many processes do not posses a reasonably small Lipschitz Bound. This is true in the case of the temporal Ozone process
shown in the left panel of Figure \ref{O3_biweekly_UPL_2004_2007.pdf}, which shows that despite the high variation in
the small scale, a slow-moving pattern is present in a larger scale. Here we provide a method to extend the
theory developed before to this case. The idea for this extension lies in the fact that in such cases the process
can be well-approximated by a function with a reasonably small LB. To formally introduce this idea, we
start by  the following definition.
\begin{definition}
Suppose $f:D \subset \R^d \rightarrow \R$. Then $f$ is said to have {\it Lipschitz Bound}, $m$, up to a {\it Bound Deviation} (BD),
$\sigma$, if there exist a function $g$ such that $|f(x)-g(x)|\leq \sigma,\;\forall x\in D$ and $g$ has Lipschitz
Bound $m$. The class of all such functions is denoted by $\ALB(D,m,\sigma)$ or by $\ALB(m,\sigma)$ when the
domain $D$ is clear from the context.
\end{definition}
Note that $f \in \ALB(m,\sigma)$ does not even need to be continuous. Therefore we have
generalized this method to functions which are not continuous but well-approximated (in the sense that $\sigma$
is small) by a continuous function $g$. Similar to the simple case with no deviation, we can define a periodic family for the case with deviations and we denote that by $
\PALB([a,b],m,\sigma)$. Note that this definition lets us include functions which are approximately periodic.
In other words, it allows a function $f$ for which $0 \leq |f(b)-f(a)| \leq \sigma$.

In the 1-dimensional case we can approximate $\LB([a,b],m)$
and $\PLB([a,b],m)$ by piece-wise linear functions with line segments with slope magnitude less than or equal to $m$ as accurately as 
desired. Let $\PL([a,b],m)$ be the set of piece-wise linear functions on $[a,b]$ with slope magnitude less than or equal to $m$ and $\PPL([a,b],m)$ be the set of periodic piece-wise linear functions on $[a,b]$ with slope magnitude less than or equal to $m$. 
Then we have the following lemma. 
\begin{lemma}
$\PL([a,b],m)$ is dense in $\LB([a,b],m)$ and $\PPL([a,b],m)$ is dense in $\PLB([a,b],m)$ in terms of sup norm: $||f||_{\infty}=\underset{x \in D}{\sup} ||f(x)||$.
\label{lemma-pl-lb-approx}
\end{lemma}
\begin{proof}
The proof can be done by considering a fine grid on $[a,b]$ and approximating the function on $[a,b]$ by a piece-wise linear function which interpolates the function values on the grid. 
\end{proof}

\begin{figure}
\centering
\includegraphics[width=0.35\textwidth]{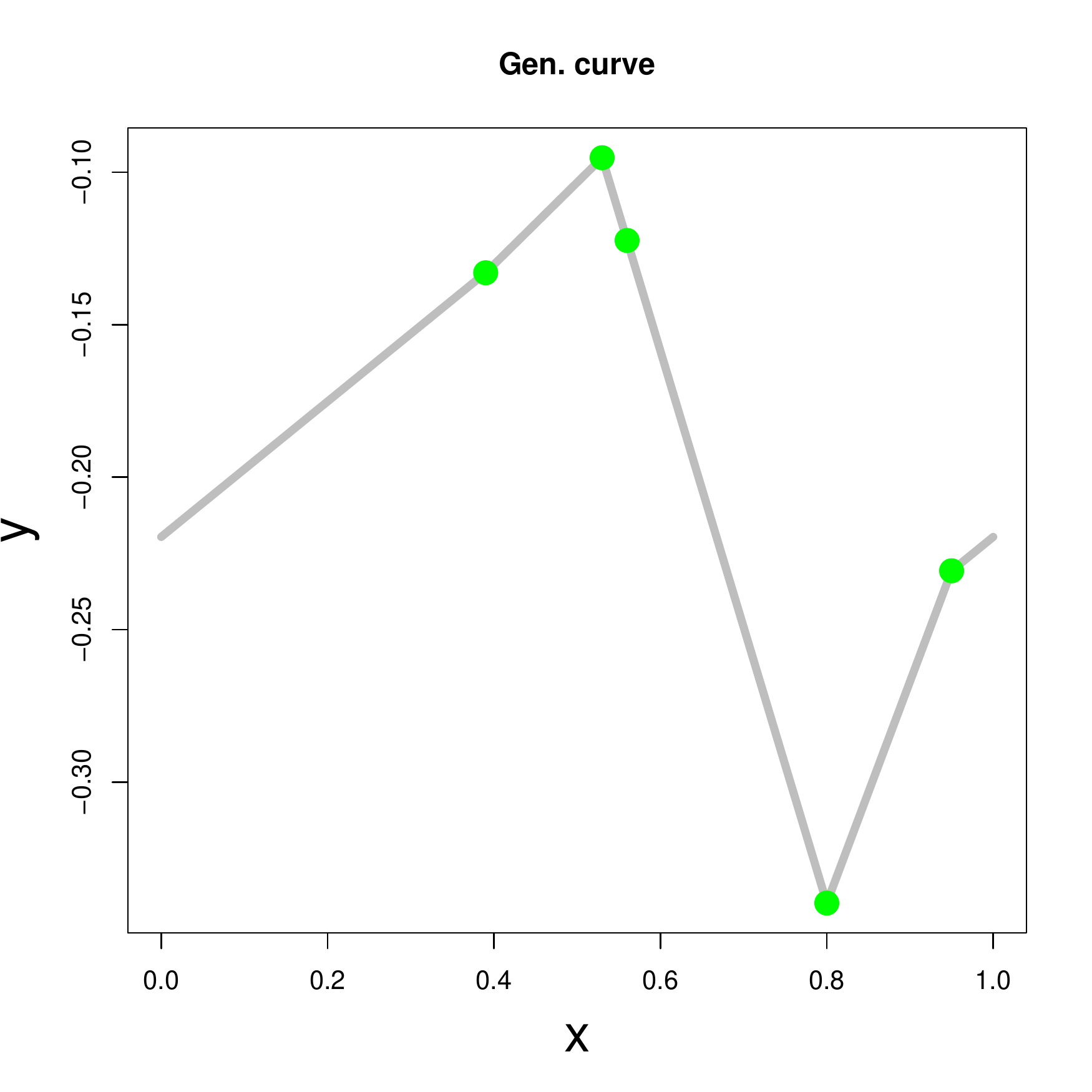}\includegraphics[width=0.35\textwidth]{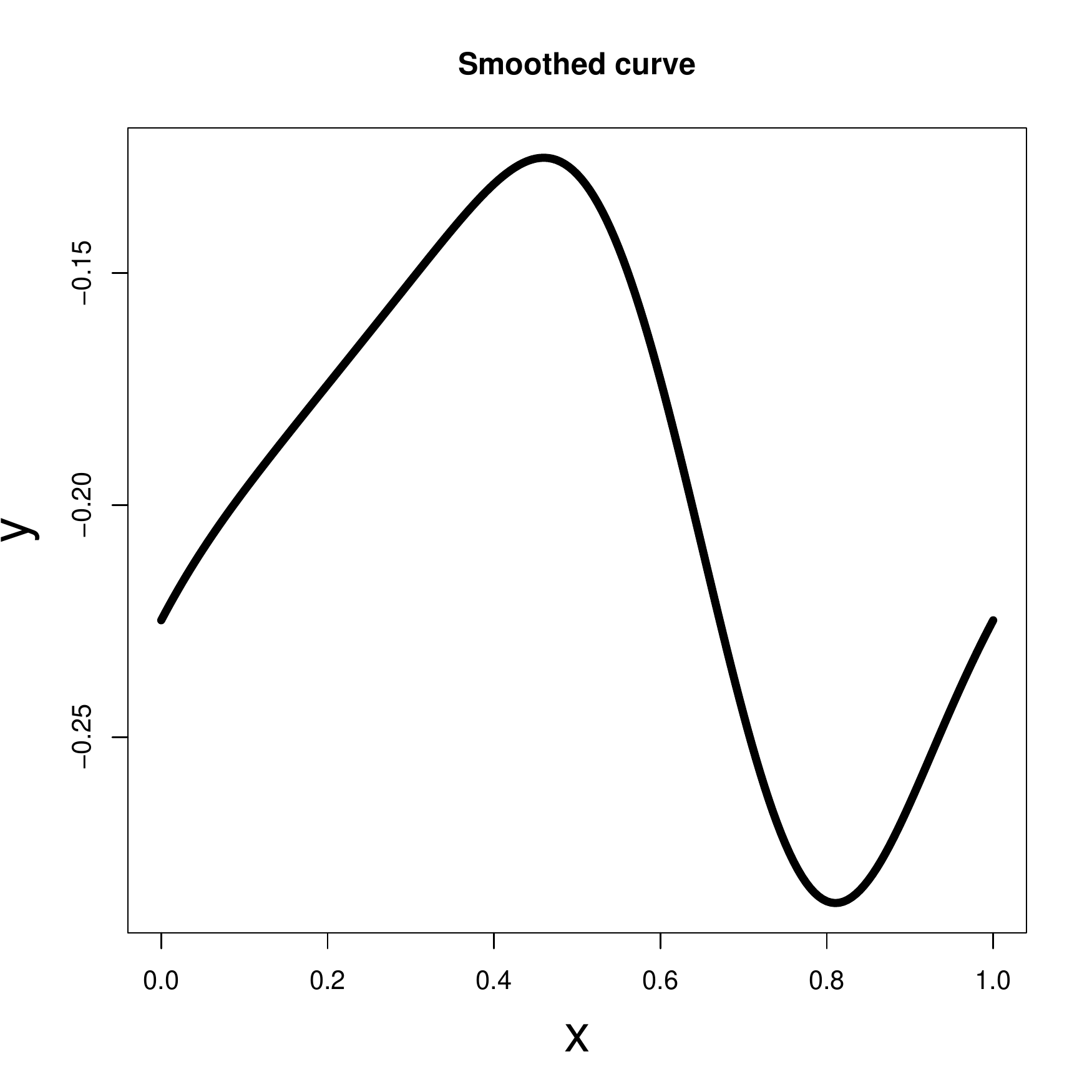}
 \caption{\small Left Panel: A simulated periodic function with 5 break points and with LB equal to 1. Right Panel: A smoothed version of the simulated curve using a moving average filter.}
 \label{per_curve_gen_5points_smooth.pdf}
\end{figure}

We can use Lemma \ref{lemma-pl-lb-approx} to simulate functions in 
$\LB([a,b],m)$ and $\PLB([a,b],m)$.
 An example of the periodic case with 5 break points $p_1,\cdots,p_5$ is given in Figure
\ref{per_curve_gen_5points_smooth.pdf} along with its more smooth version (smoothed to calculating a moving average from the left curve). In order to do several simulations we
need to define random procedures to find the break points and the slopes. Later we use uniform
distributions for both cases but we also make sure the break points are not too close as discussed later.

\subsection{Loss functions}
\label{sect:loss-functions}
We can consider various loss functions to asses the efficiency of the approximation methods.  As we discussed in the
introduction, we can consider two general types of losses: losses for approximating the integral of a function;
and losses for approximating the point-wise values of the function:
\begin{itemize}
\item The integral approximation loss: $IL(f,\hat{f}):=|\int_{D} f(x)dx-\int_{D} \hat{f}(x)dx|.$
%In fact this loss only depends on $\hat{f}$ through $\int_{D} \hat{f}(x)dx$. More precisely
% \[\int_{D} g_1(x)dx=\int_{D} g_2(x)dx \Rightarrow IL(f,g_1)=IL(f,g_2).\]
\item The point-wise approximation loss, for which two measures can be considered:
\begin{itemize}
\item[(a)] Supremum point-wise loss: $SPWL(f,\hat{f}):=\sup_{x \in D} |f(x)-\hat{f}(x)|,$
\item[(b)] Mean point-wise loss: $MPWL(f,\hat{f}):=\int_{D} |f(x)-\hat{f}(x)|/v(D),$
where $v(D)=\int_{D}1dx.$ %For example $v([a,b])=b-a$.
\end{itemize}
\end{itemize}
The error measures defined above are not scale-free and cannot inform us how much of the variation of the
function is captured using an approximation method. In order to standardize the above error, we can divide them
by the {\it diameter} of $f$, which we define to be $diam(f):=\sup_{D}(f)-\inf_{D}(f).$
Then we define the {\it standardized supremum point-wise loss} to be
\[SSPWL(f,\hat{f})=SPWL(f,\hat{f})/diam(f).\]
It is easy to see that if the approximation method is {\it scale-free}: $g=a+bf \Rightarrow \hat{g}=a+b\hat{f},$
then $SSPWL$  is also scale-free: $SSPWL(g,\hat{g})=SSPWL(f,\hat{f}).$
Any reasonable approximation method and the ones discussed here are scale-free. Also note that $0\leq SSPWL \leq
1$.

Suppose we have a family of functions with LB, $m$, and defined on $D$. We define
the {\it family-standardized supremum point-wise loss} to be
\[FSPWL(f,\hat{f})=SPWL(f,\hat{f})/v(D)m.\]
 Again it is easy to see that if the approximation method is scale-free then $FSPWL$  is also scale-free. It is
clear that $0 \leq FSPWL \leq SSPWL \leq 1$ in general. Similarly we define standardized versions of $IL,\;MPWL$
and denote them by $SIL,\;SMPWL$. We denote their family-standardized versions by $FIL,\;FMPWL$.

\subsection{Approximations and their performance}
\label{sect:interpolation}
Suppose the value of $f:D \subset \R^d \rightarrow \R$ is observed at $n$ points ${\bf x}:=(x_1,\cdots,x_n)$, where each $x_i$ is a column vector of length $d$; it  takes
 values ${\bf y}:=(y_1,\cdots,y_n)$; a LB, $m,$ is available; and we are interested in
approximating $f$ at unobserved points $x$  in $D$. We denote such an approximation by $approx[{\bf x},{\bf y}]$,
which is a function on the same domain as $f$. An {\it approximation method}, $approx[.,.],$ is formally a function
that inputs data and outputs functions: \
\begin{align*}
approx: \cup_{n=1}^n \R^{n\times d} \times \R^n &\;\;\; \rightarrow \;\;\;\;\; \R^{D},&\\
 ({\bf x},{\bf y}) & \;\;\;\mapsto \;approx[{\bf x},{\bf y}],&
\end{align*} where $n$ is the size of the data set. In the previous section,
we introduced loss measures for assessing the distance of a given curve to the target function. This cannot
directly be used to assess the performance of an approximation method, because the true function is not available
in practice. (However $SPWL$ and $IL$ introduced
 previously are useful in simulations where the true curve is known.)

In the following, we introduce approximation (prediction) error measures
which are suitable for comparing approximation methods when the target function is not available. The definition
follows by introducing an {\it ordering} on the set of all approximations and we discuss the connection of this
ordering to the error measures.
\begin{definition}
Suppose we are interested in approximating a function $f:D \subset \R^d \rightarrow \R$, which belongs to a family of
functions $\F$. For example $\F=\LB(m)$ or $\F=\PLB(m)$. Also assume $f$ is observed on ${\bf
x}=(x_1,\cdots,x_n)$ and takes values ${\bf y}=(f(x_1),\cdots,f(x_n))$.\\
(i) We define the {\it data-informed supremum point-wise error} to be:
\[DSPWE(approx,{\bf x},{\bf y})=\sup_{f \in \F, f({\bf x})={\bf y}} SPWL(f,approx[{\bf x},{\bf y}]).\]
(ii) We define the {\it supremum point-wise error} to be:
\[SPWE(approx,{\bf x})=\sup_{f \in \F} SPWL(f,approx[{\bf x},{\bf f(x)}]).\]
\end{definition}
Similarly we can define: the {\it data-informed mean point-wise error}, $DMPWE$; the {\it mean point-wise error}, $MPWE$;
the {\it data-informed integral error}, $DIE$; the {\it integral error}, $IE$. Moreover the family-standardized version of the above errors can be obtained by dividing them by $v(D)m$ and we denote them by adding ``F'' to the title, for example $FDSPWE$ is the family-standardized version for $DSPWE$ and so on.

In the sequel, we obtain both $DSPWE$ and $SPWE$ for well-known approximation methods and for the optimal method we
develop here ($Lipfit$). The difference between $DSPWE$ and $SPWE$ is that: $DSPWE$ utilizes the extra information about the
actual values of the function at the observed values to calculate the approximation error. While $SPWE$ only uses
the position of the points for which $f$ is observed: $(x_1,\cdots,x_n)$ and the approximation error is obtained
without using the actual values of the curve $(f(x_1),\cdots,f(x_n))$. Therefore $SPWE$ is useful in the sampling
phase -- when we decide where to observe the values of a function -- in which case we do not have access to the
values of the target function a priori. However when we do have access to the values of the function, we should not
discard that information in assessing the error in the approximation and that is what $DSPWE$  achieves.
\begin{comment}
To our knowledge most of these types of approximation error measures have not been considered rigorously in the literature while $SPWE$ has been considered for the
 {\it Linear Interpolation} method  e.g.\;in \cite{book-numerical-cheney} and for Lipschitz functions e.g.\;in \cite{Beliakov-2006}. We show in the sequel that
considering $DSPWE$ not only gives us a better error measure, it does make a difference in comparing various
methods. In fact, we show that the extra information about the data can guide us to choose among methods which
yield the same $SPWE$ but differ in terms of $DSPWE$.
\end{comment}

\vspace{0.5cm}

\noindent{\bf Point-wise error function}

 The error measures defined above are useful to assess the goodness of
various loss functions on their domains. Since these are defined using the whole domain, we can consider them as
overall measures of error. As we will see various cases of approximation methods have the same overall error
in some situations. However, we can compare these approximation methods point-wise by considering the
{\it point-wise error function}, denoted by $pef$ and defined as follows:
\[pef[approx,{\bf x}, {\bf y}](x)=\sup_{f \in \F, f({\bf x})={\bf y}} |f(x)-approx[{\bf x},{\bf y}](x)|.\]
Note that $pef$ is a function on $D$ in contrast to $DSPWE$ and in fact
\[DSPWE[approx,{\bf x}, {\bf y}]=\sup_{x\in D}pef[approx,{\bf x}, {\bf y}](x).\]

When comparing two approximation methods $approx_1,approx_2$, to show the superiority of $approx_1$ to $approx_2$
(in terms of point-wise error), one ideally wants to show a superiority everywhere on the domain:
\begin{equation}
pef[approx_1,{\bf x}, {\bf y}](x) \leq pef[approx_2,{\bf x}, {\bf y}](x),\;\forall x \in D,
\label{eqn-everywhere-sup}
\end{equation}
from which we can conclude:
\[DSPWE[approx,{\bf x}, {\bf y}] \leq DSPWE[approx_2,{\bf x}, {\bf y}];\]
\[DMPWE[approx,{\bf x}, {\bf y}] \leq DMPWE[approx_2,{\bf x}, {\bf y}].\]
We denote the relation in Equation \ref{eqn-everywhere-sup} by $approx_1 \preceq_{pw} approx_2$. If both
$approx_1 \preceq_{pw} approx_2$ and $approx_2 \preceq_{pw} approx_1$ hold, we write $approx_1 =_{pw} approx_2$.
Also note that this relation is {\it transitive} and {\it anti-symmetric}.

However it is not a {\it total ordering} in general. And there may exist  a pair  $approx_1,approx_2$ for which neither $approx_1
\preceq _{pw}approx_2$ nor $approx_2 \preceq_{pw}approx_1$ is true. In contrast, if we define an ordering using
$DSPWE$  or $DIE$, then clearly we get a total ordering (since the usual ordering of real numbers is a total ordering).
Interestingly, this is one of those rare situations that although the ordering is not a total ordering, there is a solution to the approximation problem
which minimizes the approximation error in terms of $\preceq_{pw}$.

In the case of integral losses, we have a similar contrast between $DIE$ and $IE$ as before: $DIE$ is a
better measure of error, while $IE$ is useful when the values of the function are not available. However, for approximating integrals, there is no point-wise error function version
similar to $pef$,  since by definition we are interested in integrals.

Below we formally define various
approximation methods. Important examples of the (1-d) approximation methods are:
\begin{enumerate}
\item {\bf Average ($AVG$)}: $AVG[{\bf x},{\bf y}](x):=\frac{1}{n}\sum_{i=1}^n {f(x_i)}$. 

\item {\bf Nearest Neighbor ($NN$)}: to every point, assigns the value of $f$ at the closest available point.
\item {\bf Periodic Nearest Neighbor ($PNN$)}: This is a variation of $NN$ to use the assumed periodicity in $f$. We add two
points to $x_1,\cdots,x_n$: $x_n^*=x_n-(b-a)$ and $x_1^*=x_1+(b-a)$. Note that by the periodicity assumption
$f(x_1^*)=f(x_1)$ and
    $f(x_n^*)=f(x_n)$. Then we apply the $NN$ method to the points
    \[(x_n^*,f(x_n^*)),(x_1,f(x_1)),\cdots,(x_n,f(x_n)),(x_1^*,f(x_1^*)).\]
%To our knowledge this method has not been considered before.
\item {\bf Linear Interpolation ($LI$)}: This method draws line segments between each pair of points
$(x_i,f(x_i)),(x_{i+1},f(x_{i+1})),\;i=1,\cdots,(n-1)$ and approximates $t \in [x_i,x_{i+1}]$ by the corresponding
value on the line. For points $[a,x_{1}]$ and $[x_{n},b]$, we assign the nearest neighbor value.
\item {\bf Periodic Linear Interpolation ($PLI$)}:  This can be defined similarly to $PNN.$ 

\item {\bf Regression} 

\item {\bf Regularization:} Examples of regularization methods are smoothing splines and LOESS.
The $Lipfit$ method developed in this work can also be considered as a regularization method as we discuss later.
\end{enumerate}
The above approximation methods, except for the last two, are data-range faithful, i.e.\;the approximated function is in the range of the data. While, regression methods and the regularization methods in general are not data-range faithful,
$Lipfit$ which can also be viewed as a regularization method is data-range faithful. 
Also the above approximation methods can be immediately extended to multidimensional input space, except for $LI$ and the periodic cases.

\begin{figure}
\centering
\includegraphics[width=0.33\textwidth]{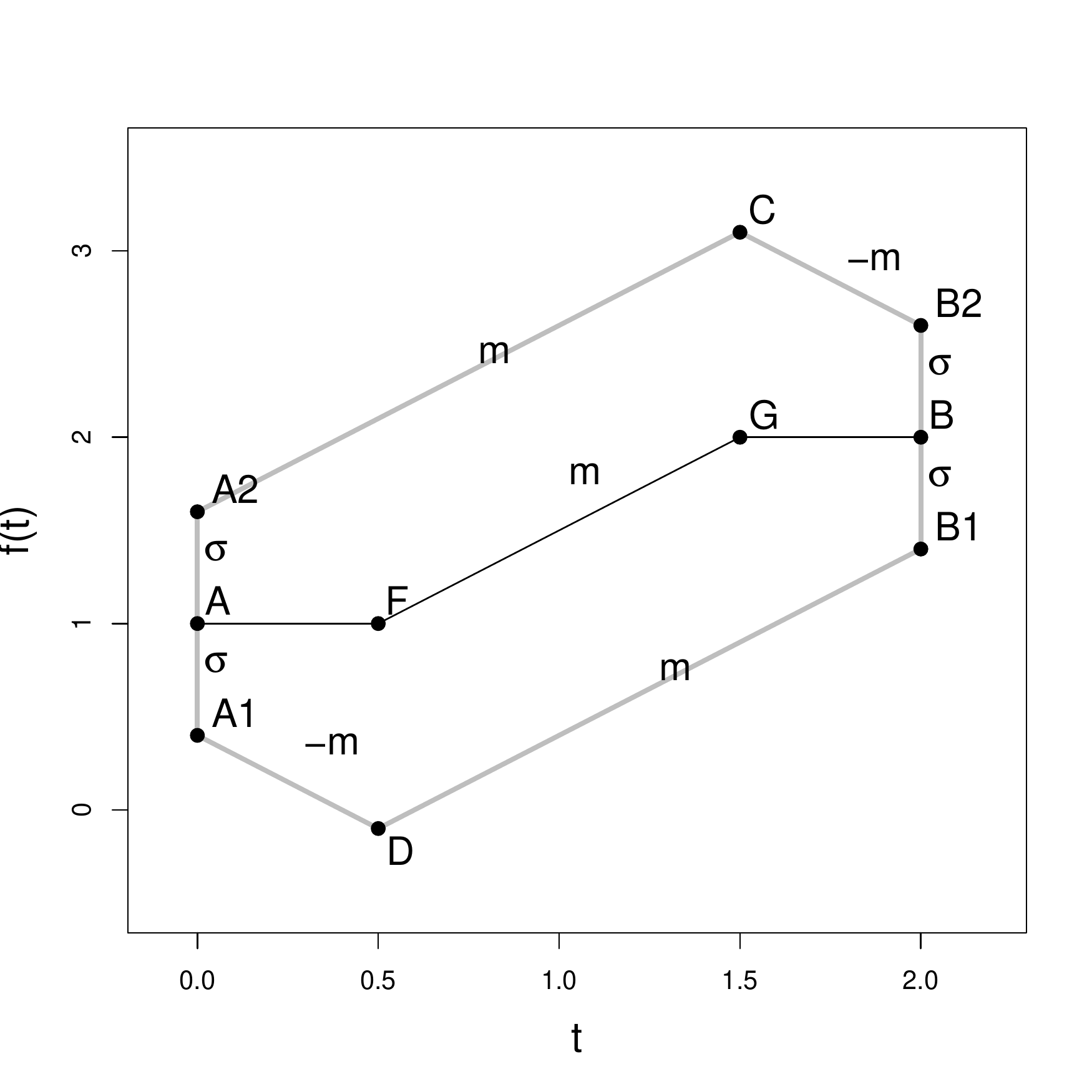}\includegraphics[width=0.33\textwidth]{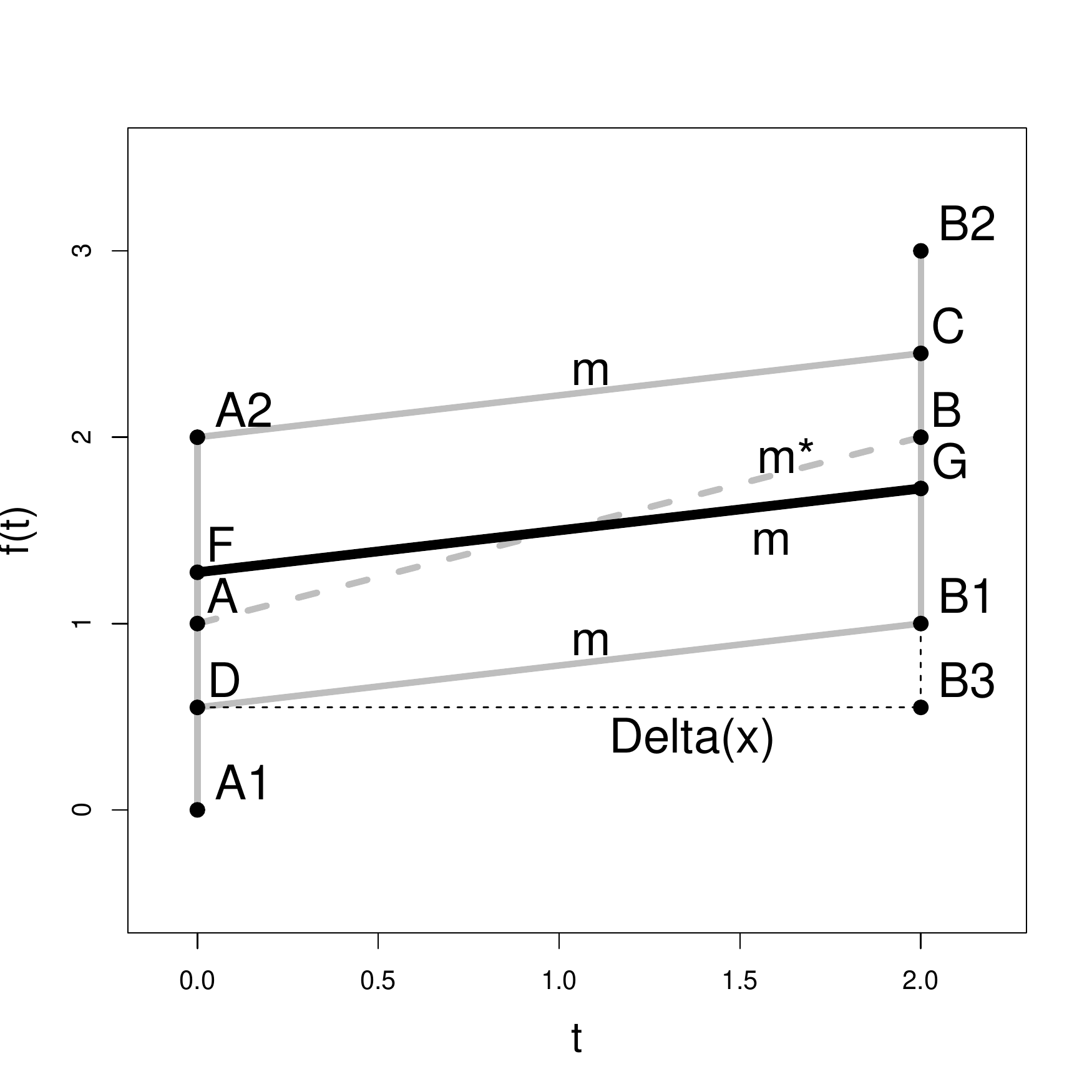}
\includegraphics[width=0.33\textwidth]{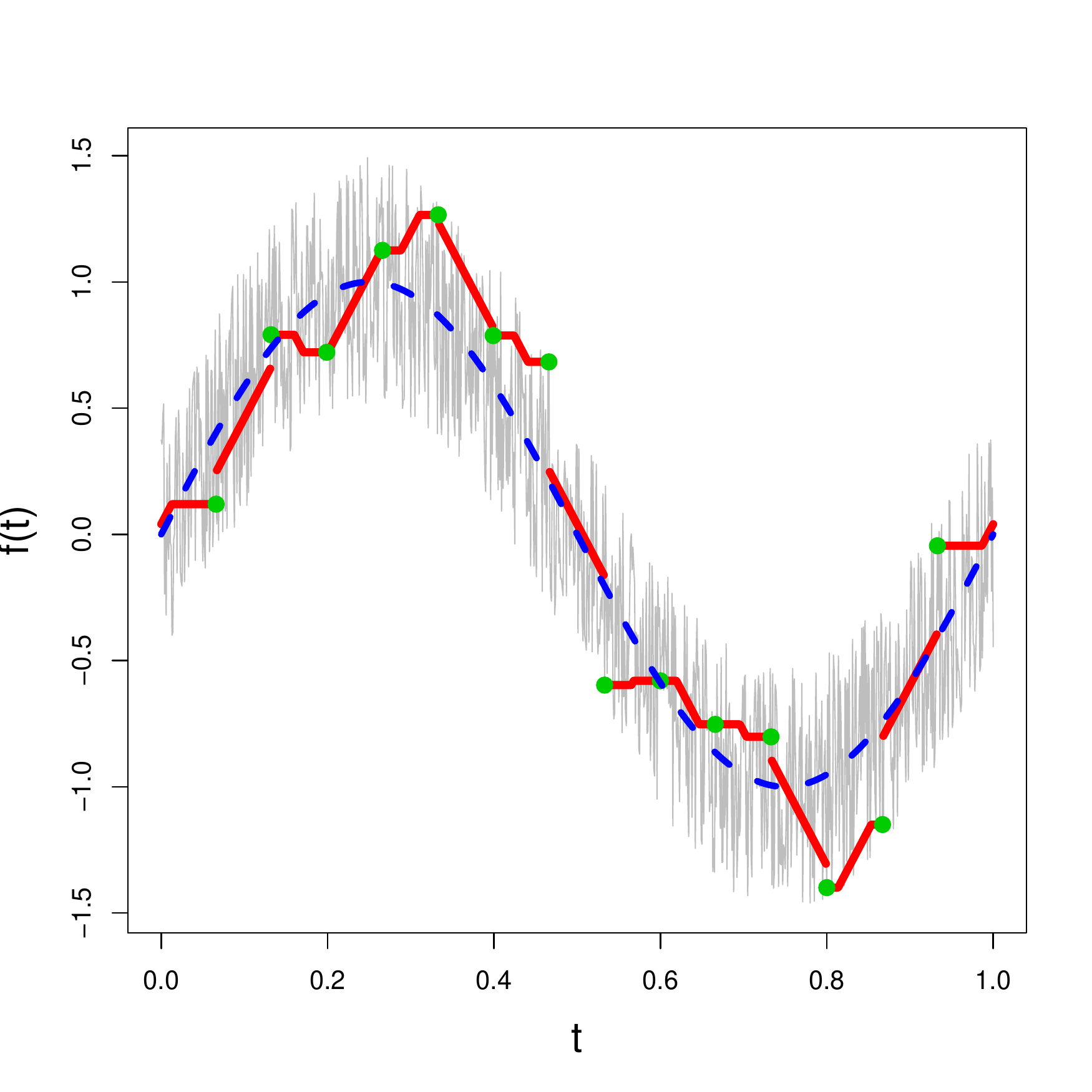}
 \caption{\small Left Panel: $Lipfit$ method with deviation ($|m^{\star}| \leq m$, Case 1). Middle Panel: $Lipfit$ method with deviation,
  ($|m^{\star}| > m$, Case 2). Right Panel: The curve (grey) is generated by adding uniform deviations from the interval $[-\sigma,\sigma]$ to a curve with given LB. The data (filled circles) are fitted with  $Lipfit$ method (thick curve). }
 \label{opt_proof_with_error.pdf}
\end{figure}

\vspace{0.5cm}

\noindent{\bf Lipfit for (1-d) functions:}
Suppose $f:[a,b] \rightarrow \R,$ belong to a $\ALB(m,\sigma)$ and its trajectory contains the two points $A=(x_A,y_A)$ and 
$B=(x_B,y_B)$: $f(x_A)=y_A,\;f(x_B)=y_B,$ where $x_A,x_B$ are two consecutive points in the data available to us (no data is available between $x_A$ and $x_B$).   Also define  $A_1=(x_A,y_A-\sigma), A_2=(x_A,y_A+\sigma)$ and $B_1=(x_B,y_B-\sigma), B_2=(x_B,y_B+\sigma)$. From each of $A_1,A_2,B_1,B_2$ draw lines with slopes $m,-m$. Let $m^{\star}=(y_B-y_A)/(x_B-x_A)$ be the slope of the line segment $AB$.
\begin{itemize}
\item Case 1:  If $|m^{\star}|\leq m$ then the method is defined as before (Figure \ref{opt_proof_with_error.pdf}, left panel).\\
These lines will intersect to form the 6-sided polygon $A_1A_2CB_2B_1D$. Then consider the functions: $H^{upper}$ on $[x_A,x_B]$ be the function which goes along the line segments $A_2C$ and $CB_2$; $H^{lower}$ on $[x_A,x_B]$ be the function which goes along the line segments $A_1D$ and $DB_1$. The $Lipfit$ on $(x_A,x_B)$ is equal to $(H^{upper}+H^{lower})/2$   and goes along the line segments $AF$, $FG$, $GB$.\\
 \begin{itemize}
 \item Thus the $Lipfit$ solution can be identified as follows.\\  Let $\Delta=(\Delta_x-|\Delta_y/m|)/2,$  $F=(x_A+\Delta,y_A)$
 and $G=(x_B-\Delta,y_B)$. Then $Lipfit$ goes along the line segments $AF$, $FG$, $GB$.
\end{itemize}
\item Case 2: If $|m^{\star}|> m$ as shown in Figure
\ref{opt_proof_with_error.pdf} (middle panel).\\
These lines will intersect to form the parallelogram, $DA_2CB_1$. Then consider the functions: $H^{upper}$ on $[x_A,x_B]$ be the function which goes along the line segment $A_2C$; $H^{lower}$ on $[x_A,x_B]$ be the function which goes along the line segment $DB_1$. The $Lipfit$ on $(x_A,x_B)$ is equal to $(H^{upper}+H^{lower})/2$   and goes along the line segment $FG$.\\
\begin{itemize}
\item Thus the $Lipfit$ solution can be identified as follows.\\ 
Define $\Delta'=\Delta_x(|m^{\star}|-m)/2$ and the points
\[F=(x_A,y_A+sign(m^{\star}))\Delta',\;\;G=(x_B,y_B-sign(m^{\star}))\Delta'.\]
Then the $Lipfit$ method is given by the line segment $FG$.
\end{itemize}
\end{itemize}
In the above we observe that  the value of BD is not needed for applying the $Lipfit$ method. The $Lipfit$ method is data-range faithful, i.e.\;the approximated curve is in the range of the
data. Moreover, the $Lipfit$ method is an interpolation method in the sense that on
the observed points returns the observed values. However it may have sudden discontinuity at the observed values.The $Lipfit$ method can also be generalized to multidimensional case easily as we discuss below.

\vspace{1cm} \noindent {\bfseries Lipfit for general (multidimensional) case:}
Suppose a function $f$ is given at ${\bf x}=(x_1,\cdots,x_n)$ where each $x_i$ is a column vector of length $d$ denoting a point in $D \subset \R^d$, with values equal to
${\bf y }=(y_1,\cdots,y_n)$. Then suppose we are interested to approximate $f$ at a  point  $x \in D$. Applying the Lipschitz Bound to $x$ and $x_i$ for $i=1,\cdots,n$, we get
\[|f(x)-f(x_i)| \leq m|| x-x_i || + \sigma \Rightarrow   f(x_i) -  m|| x-x_i || - \sigma \leq f(x) \leq f(x_i) +  m|| x-x_i || + \sigma,\]
from which we conclude
\[H^{lower}(x) \leq f(x) \leq H^{upper}(x),\]
where
\begin{eqnarray*}
H^{lower}(x)=\underset{i=1,\cdots,n}{\max}(f(x_i) -  m|| x-x_i ||) + \sigma(1-1_{\{x_1,\cdots,x_n\}}(x)),\\
H^{upper}(x)=\underset{i=1,\cdots,n}{\min}(f(x_i) +  m|| x-x_i ||) - \sigma(1-1_{\{x_1,\cdots,x_n\}}(x)),
\end{eqnarray*}
where $1_{\{x_1,\cdots,x_n\}}(x)=1$ if $x$ is an observed point and zero otherwise.
Then optimal solution which minimizes $|f(x)-\hat{f}(x)|$ is given by
\[\hat{f}(x)=(H^{lower}(x)+H^{upper}(x))/2.\]
Note that $\sigma$ cancels out and we see that for the general solution it also does not appear in the solution (as it was the case for the 1-d case we discussed before).
In order to see the optimality, it is sufficient to note that: (1) both $H^{lower}(x),H^{upper}(x)$ interpolate the data; (2) they belong to $\ALB(m,\sigma)$; (3) it is impossible for $f(x)$
to be outside the range $[H^{lower}(x),H^{upper}(x)]$. Thus we have extended the results in \cite{Sukharev-1978} and \cite{Beliakov-2006} and the $Lipfit$ method is the optimal method in terms of each of $DSPWE, SPWE, DIE, IE$.
  Also note that our solution in 1-d case match this solution and is computationally fast. The point-wise error function can be expressed in terms of $H^{lower}(x),H^{upper}(x)$:
\[pef(x)=(H^{lower}(x)-H^{upper}(x))/2,\]
from which $DSPWE$ and $DIE$ can be calculated for the multidimensional case. For the 1-d case $DSPWE, SPWE, DIE, IE$
can be found in closed-form and are given in the following theorem.

\begin{theorem}
\label{theo-err-bounds-ALB} Suppose a function belongs to $\ALB(m,\sigma)$ with trajectory going through points
$A=(x_A,y_A)$ and $B=(x_B,y_B)$. Denote the slope of the line from $A$ to $B$ by $m^{\star}$. Also define
$\Delta_x= (x_B-x_A),\;\Delta_y=(y_B-y_A),$ and $\Delta=(\Delta_x-|\Delta_y/m|)/2$. Then the data-informed supremum point-wise error,
($DSPWE,$) and supremum point-wise error, $SPWE$, for various methods to approximate the curve
on $[x_A,x_B]$ are given below.
\begin{itemize}
\item  If $|m^{\star}|\leq m$:
\begin{itemize}
\item $DSPWE[NN,(x_A,x_B),(y_A,y_B)]=m{\Delta_x}/{2}+\sigma$.
\item $DSPWE[LI,(x_A,x_B),(y_A,y_B)]=\Delta(m+|m^{\star}|)+\sigma$.
\item $DSPWE[Lipfit,(x_A,x_B),(y_A,y_B)]=\Delta m+\sigma$.
\item $DIE[approx,(x_A,x_B),(y_A,y_B)]=\frac{m^2-{m^{\star}}^2}{4m}(\Delta_x)^2+\sigma \Delta_x,\;approx=NN,LI,Lipfit.$
\end{itemize}
\item  If $|m^{\star}|>m$, define $\Delta'=\Delta_x(|m^{\star}|-m)/2$:
\begin{itemize}
\item $DSPWE[NN,(x_A,x_B),(y_A,y_B)]=m{\Delta_x}/{2}+\sigma$.
\item $DSPWE[LI,(x_A,x_B),(y_A,y_B)]=\sigma$.
\item $DSPWE[Lipfit,(x_A,x_B),(y_A,y_B)]=\sigma-\Delta'$.
\item $DIE[approx,(x_A,x_B),(y_A,y_B)]=(\sigma-\Delta')\Delta_x,\;approx=NN,LI,Lipfit.$
\end{itemize}
\item $SPWE[approx,(x_A,x_B)]=m{\Delta_x}/{2}+\sigma,\;approx=NN, LI, Lipfit$.
\end{itemize}
\end{theorem}
\begin{proof}
The proof can be done geometrically and the idea of the proof is given in Figure \ref{opt_proof_with_error.pdf}. For a detailed proof see \cite{hosseini-2013-sparse-tech}.
\end{proof}

For each of the $LI$ and $NN$ methods, we introduced periodic versions: $PLI$ and $PNN$ respectively and the same
can be done for the $Lipfit$ method which we denote by $PLipfit$. Again it is true that $PLipfit$ is uniquely optimal when
$n$ points are available for a periodic function in terms of $\preceq_{pw}$ and therefore in terms of $DSPWE$ as
well as $DMPWE$.

\vspace{0.5cm}

\section{Simulation studies}
\label{sect:simulation} This section uses simulations to investigate approximating functions in the framework
developed in this paper for the 1-dimensional domain case. Lemma \ref{lemma-pl-lb-approx} is used here for simulating appropriate functions to compare approximation methods in the 1-d case.
  \cite{hosseini-2013-sparse-tech} compare the methods for the case without deviation and when the deviation is negligible, showing that the $Lipfit$ method is superior to other methods (LI, NN, regression) when the data are sparse in terms of $DSPWE, DIE$. Moreover it is shown that if the functions being simulated are periodic, the periodic version of the methods $PLI, PNN, PLipfit$ improve the prediction error (integral or point-wise) significantly with $PLipfit$ reaching the smallest errors. Here  we perform simulations
for functions that are generated with a given Lipschitz Bound and a non-negligible deviation which is generated from a uniform distribution. This is a special case, because in general the deviations can also have some remaining patterns. However even in this special case, we show that the performance of different approximation methods depend on the magnitude of the deviation and the data sparsity structure. Some remaining work in this area include the multidimensional domain case and the case with more complex deviations which are beyond the scope of this paper.

\subsection{The effect of BD magnitude on method performance} \label{sect:estimating-m-eps}
Here we perform some simulations to study the effect of the magnitude of the Bound Deviation  (BD) on the
method performance. We compare these methods: (1) $Lipfit$ with the same $LB$ the curves were simulated from; (2) $Lipfit$ with  LB
larger than the one the curves were simulated from (denoted by $Lipfit.big$); (3) $Lipfit$ with LB smaller
than the one the curves were simulated from (denoted by $Lipfit.sm$); (4) $LI$; (5) regularization/regression methods
such as $LOESS$ and smoothing splines (e.g.\;see \cite{book-loess-cleveland} and \cite{book-hastie-tib}).

For simulating the curves, we pick LB=10 with 5 break points. Also the distance between each pair of the break
points is taken to be at least 1/(5+2). For the sampling scheme, we consider a sparse data case: From each of
[0,1/4], (1/4,1/2], (1/2,3/4], (3/4,1], we take 2 points uniformly at random. Therefore there are 8 points
available from [0,1] and there is some assurance to cover the whole interval due to the sampling scheme. In
contrast to the previous simulations for which we assume BD to be very small, here we consider larger BDs to
study the effect of its magnitude. Figure \ref{Lipfit_method_with_err_simulation_sig1.pdf} depicts the fits of
various methods to the 8 points for one out of 1000 simulations for BD=1.

To investigate the method performance dependence on the BD magnitude, we consider two cases: Case 1, BD=0.5;
Case 2, BD=1.5. Figure \ref{compare_methods_sig_half.pdf} presents $(25\%, 50\%, 75\%)$ quantiles of the $MPWL$
for the methods (1) $Lipfit$; (2) $Lipfit.big$; (3) $Lipfit.sm$; (4) $LI$; (5) $LOESS$. 

Now we summarize the
results. In Case 1, where BD=0.5 and smaller than Case 2, we observe that the methods $Lipfit$ and $LI$ perform
almost equally well and outperform the other methods. In Case 2, in contrast to Case 1, we observe that $Lipfit$ and $Lipfit.sm$ perform almost equally well,
outperforming $LI$ in particular. In both cases $Lipfit.big$ performs poorly since assuming a too big $Lipfit$ will make the approximation tend to the $NN$ method which is a poor method. The intuition that $LI$ is performing better in contrast to $Lipfit.sm$ in Case 1 and this is reversed in Case
2 is as follows: In Case 1 the BD is relatively small and therefore joining the available points using the $LI$
method does not introduce a large approximation error; while in Case 2 it could introduce a large error. In
contrast $Lipfit.sm$ works by moderating the slope of the joining line between two available data points too
aggressively -- especially in Case 1 -- believing much of the slope is due to the deviation (BD) rather than a pattern
(LB). This is because $Lipfit.sm$ is supposing a much smaller LB, (LB=1), than the one the curve was generated from (LB=10). Finally note that if we choose smaller number of data points, for example $5$ points, similar results are
obtained and the LOESS method inferiority to the other methods is magnified. We do not include those simulations
here for brevity.
\begin{figure}
\centering
\includegraphics[width=0.55\textwidth]{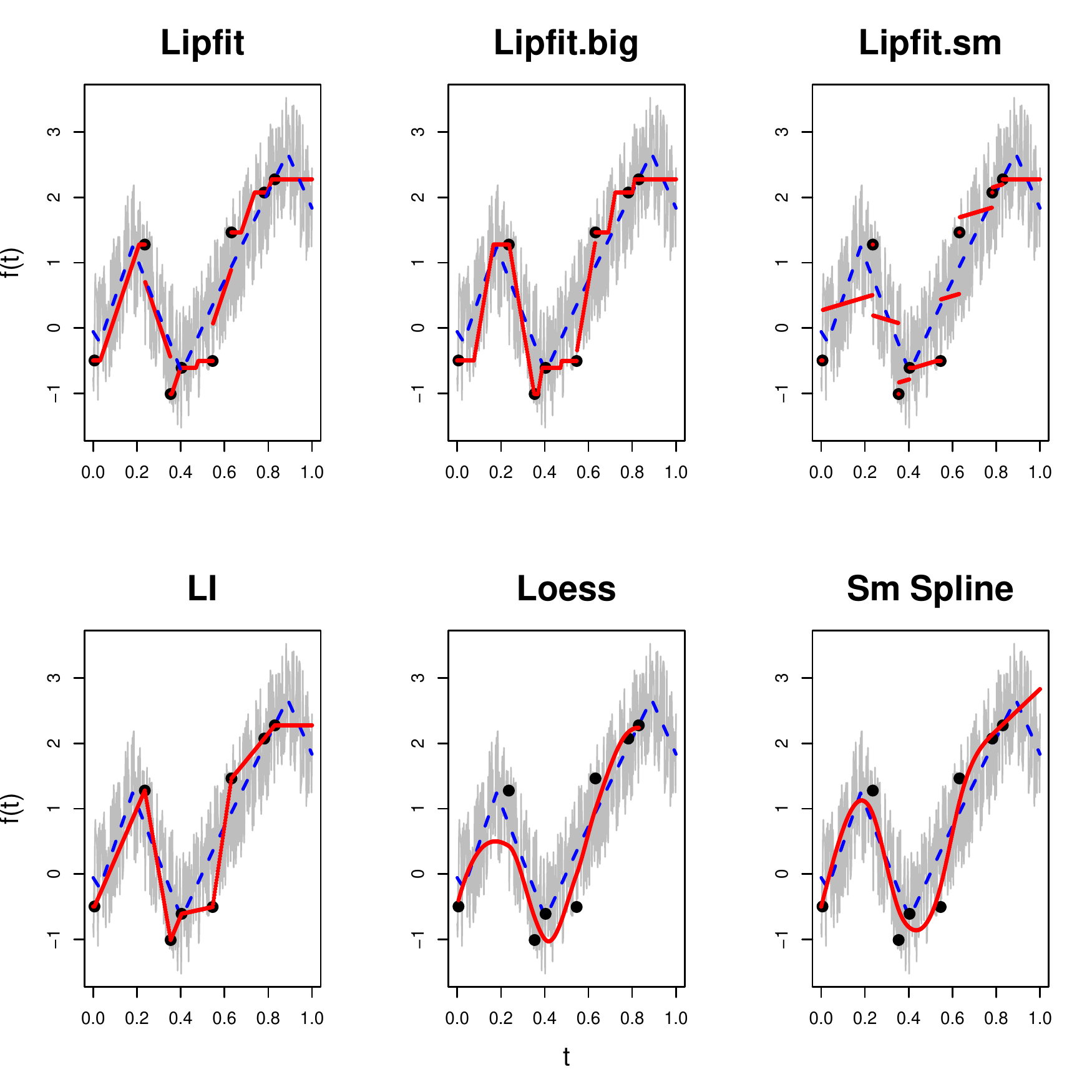}
 \caption{\small The fits using various methods using 8 available points are given where the
 target function is given grey and the fits are given in dark. The deviation-free generated curves are given with dashed lines. The simulations are done by using 5 break points and
 LB=10,\;BD=1.}
 \label{Lipfit_method_with_err_simulation_sig1.pdf}
\end{figure}
\begin{figure}
\centering
\includegraphics[width=0.3\textwidth]{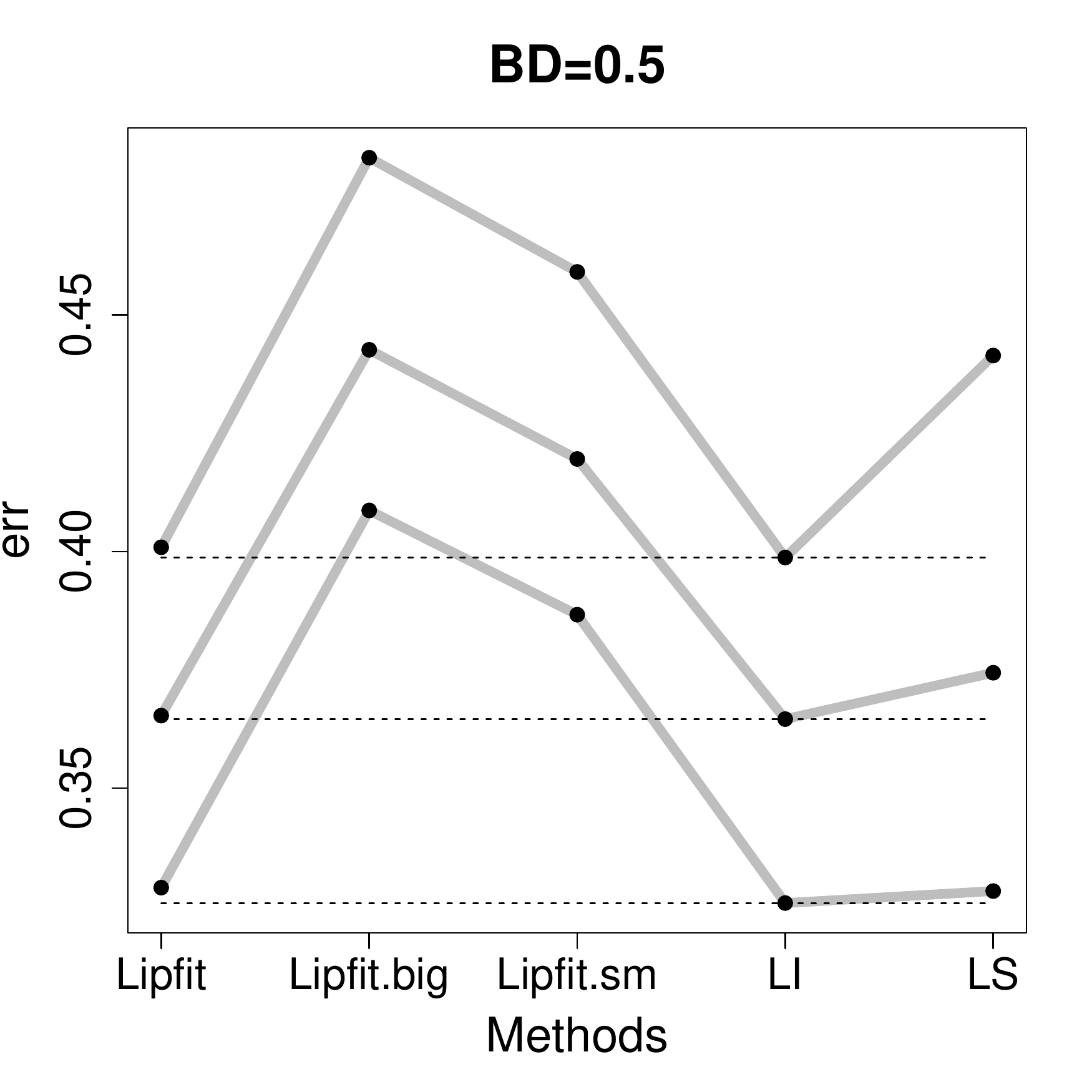}\includegraphics[width=0.3\textwidth]{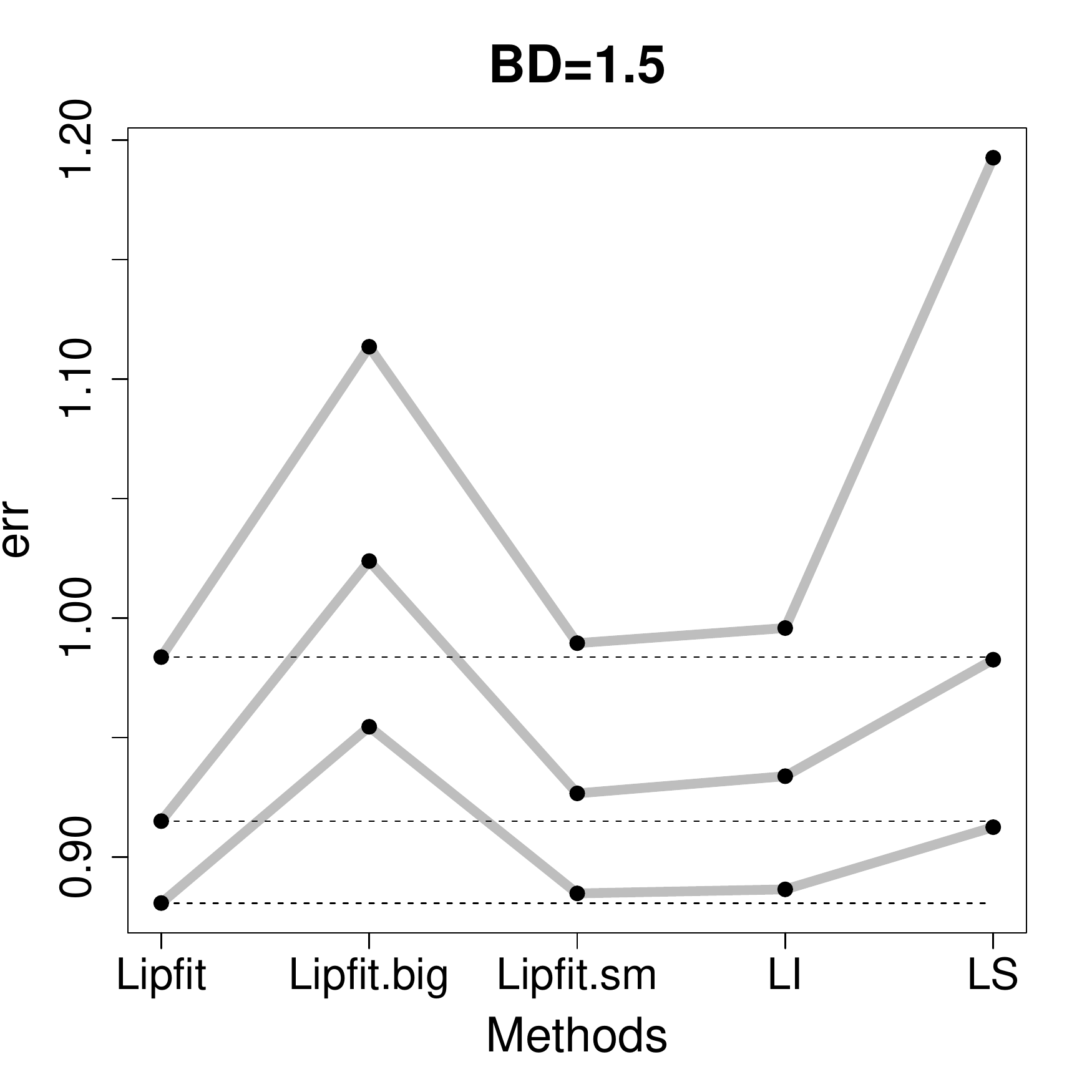}
 \caption{\small The plots depict the $(25\%,50\%,75\%)$ quantile of the MPWL
 by various method for two cases: left panel is for BD=0.5; right panel is for BD=1.5}
 \label{compare_methods_sig_half.pdf}
\end{figure}
\subsection{The effect of data sparsity on method performance} \label{sect:estimating-m-eps}
This subsection further investigates the data size and sparsity effect on the method performance. In the previous
sections, we showed that when the data are sparse over all the interval of interest, the $Lipfit$ method performs
better in contrast to $NN,\;LI$ and standard smoothing methods. We also studied the effect of using a too big
LB or too small LB in the data sparse case and for various magnitudes of error. Here we consider two new
cases: (1) The data is dense over all the interval; (2) the data size is large however, the data is sparse in
some sub-intervals due to non-uniformity of the data locations. We call such data ``locally sparse''. For both
cases, we simulate curves with 5 break points and with LB=10,\;BD=0.5.
(Case 1): From each of the intervals [0,1/4],(1/4,1/2],(1/2,3/4],(3/4,1], we take 10 points uniformly at random.
(Case 2): From each of the interval [0,1/4], we sample 50 points uniformly at random and only one point from
each of the intervals (1/4,1/2],(1/2,3/4],(3/4,1], uniformly at random.

Figure \ref{compare_methods_sig_half_ss40.pdf} depicts the error quantiles for the two cases and we summarize the
results as follows. In the data dense case (Case 1, left panel), the smoothing method ($LOESS$) has performed optimally for the lower
quantiles but still inferior to the $Lipfit$ method with the correct or small LB in higher quantiles.
 In the locally sparse data case (Case 2, right panel), we observe that the result is almost identical to the
data sparse case over the entire interval. In other words having a large data set is not necessarily going to
change the results if the data is still sparse in large sub-intervals and the smoothing methods will continue to
perform poorly.
\begin{figure}
\centering
\includegraphics[width=0.3\textwidth]{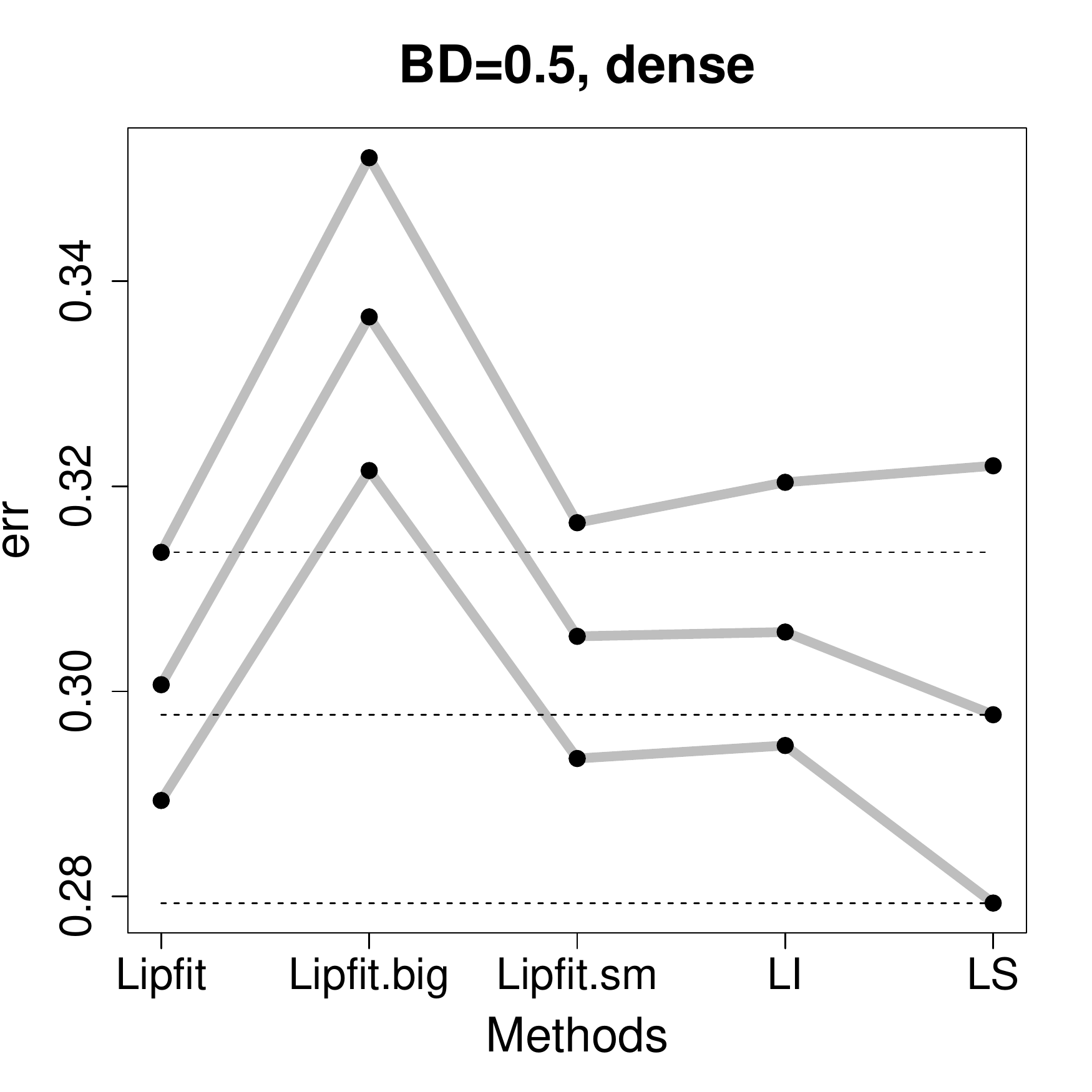}\includegraphics[width=0.3\textwidth]{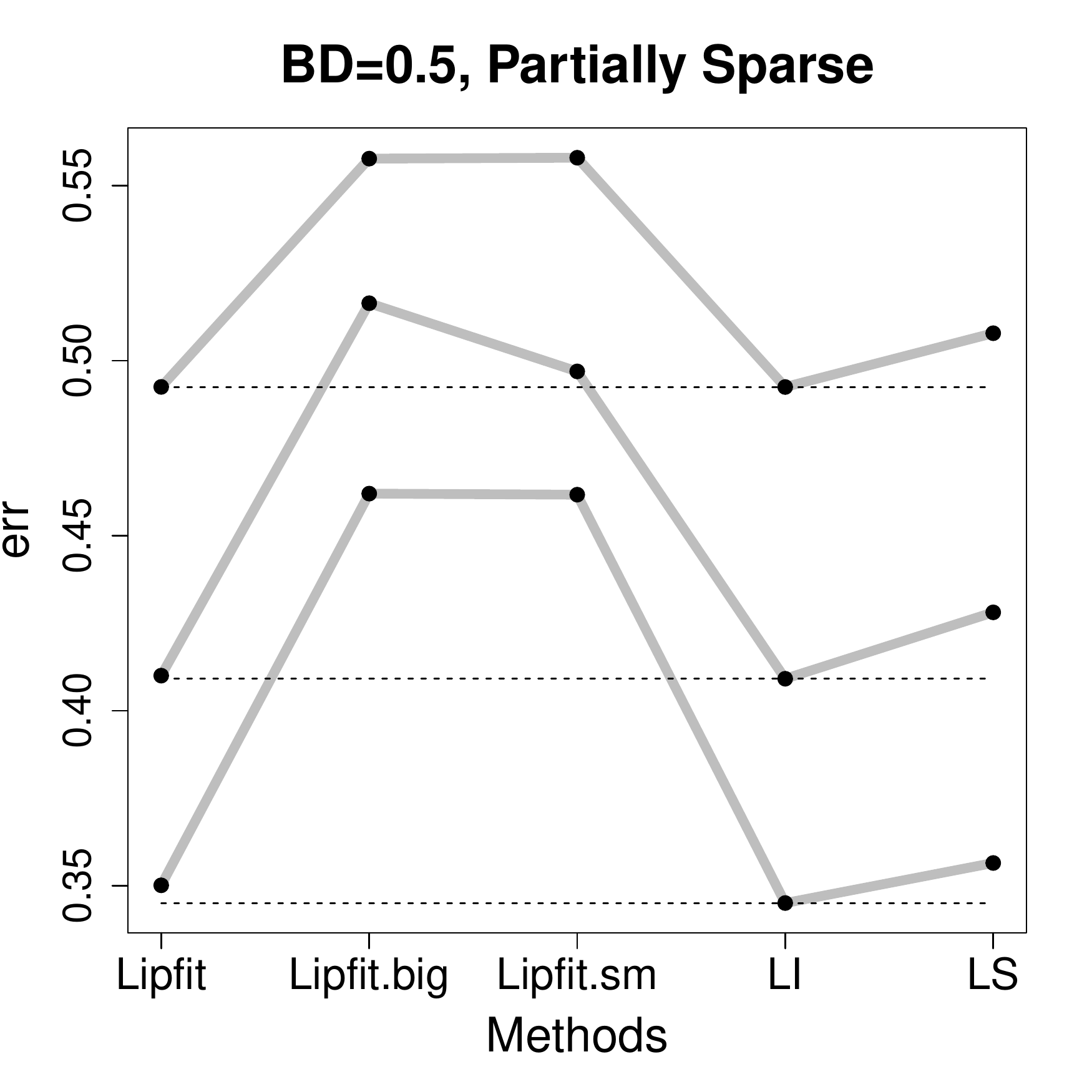}
 \caption{\small The figure depicts $(25\%,50\%,75\%)$ quantiles of the MPWL for two cases. (Left Panel) $BD=0.5$ and dense data of size 40.
 (Right Panel) $BD=0.5$ with data size 53 and non-uniform data, dense in some interval $[0,1/4]$ with 50 data points
  and one data point in each of (1/4,1/2],(1/2,3/4],(3/4,1].}
 \label{compare_methods_sig_half_ss40.pdf}
 \end{figure}

\section{The variation-deviation (LB-BD) trade-off}
\label{sect:LB-BD-trade-off}
For a given function $f:D \rightarrow \R$, LB and BD are not unique. In fact for any BD, $\sigma \in
\R^{\geq 0}$ one can find, LB, $m \in \R^{\geq 0}$ (non-negative numbers) such that $f \in \ALB(m,\sigma)$. This is
the motivation for the following definition.
\begin{definition}
Suppose $f:D\subset \R^d \rightarrow \R$. Then the LB-BD function (curve) associated with $f$ -- denoted by $\gamma_f$ -- is defined as follows: $\gamma_f:\R^{\geq 0} \rightarrow \R^{\geq 0};$
\[\gamma_f(m)=\inf \{\sigma\; |\; f \in \ALB(D,m,\sigma)\}.\]
\end{definition}
We can also consider an ``inverse" for $\gamma_f,\; \gamma^{-1}_f:\R^{\geq 0} \rightarrow \R^{\geq 0};$
\[\gamma^{-1}_f(\sigma)=\inf \{m\; |\; f \in \ALB(D,m,\sigma)\}.\]
We call the inverse also the LB-BD curve by slight abuse of naming. In Figure \ref{sin_LB_BDE_curve_CVX_analytic} (Right Panel) the LB-BD curve for
the function $f(x)=\sin(2\pi x)$ is given. The LB for $f$ is equal to $2\pi$. However if we allow for a deviation of $\sigma$, as depicted by the grey curves (Left Panel), there is
a function inside the area defined the grey curves which has a smaller LB. We can also define a LB-BD curve for the periodic case as follows.
\begin{definition}
Suppose $f:[a,b] \rightarrow \R$. Then the periodic LB-BD curve associated with $f$, 
$\gamma_f^p:\R^{\geq 0} \rightarrow \R^{\geq 0},$ is defined as follows:
\[\gamma_f^p(m)=\inf \{\sigma\; |\; f \in \PALB([a,b],m,\sigma)\}.\]
\end{definition}
In the following lemma, we give the LB-BD curve for some simple functions.
\begin{lemma}
Below we give the LB-BD curve, $\gamma_f^{-1}$, for various functions $f:[a,b] \rightarrow \R$.
\begin{itemize}
\item[(a)] $f(x)=mx:\;\;\;\gamma_f^{-1}(\sigma)=\max\{0,m-2\sigma/(b-a)\}.$
\item[(b)] $f(x)=m|x-(b-a)/2|:\;\;\;\gamma_f^{-1}(\sigma)=\max\{0,m-4\sigma/(b-a)\}.$
\item[(c)] $f(x)=\sin(2\pi x)$ and $[a,b]=[0,1]$ then $\gamma_f^{-1}(\sigma)=|2\pi\cos(2\pi x_B)|,$ where
$x_B$ is the unique solution of the equation (Figure \ref{sin_LB_BDE_curve_CVX_analytic}):
\begin{equation}
\sin(2 \pi x_B) - 2\pi(x_B-1/2)\cos(2\pi x_B)-\sigma=0,\; 1/4 \leq x_B \leq 1/2. \label{eqn-sin-LB-BD}
\end{equation}
 
\end{itemize}
\label{lemma-LB-BD-examples-analytic}
\end{lemma}

\begin{proof} (Lemma \ref{lemma-LB-BD-examples-analytic})
(a) and (b) are easy to show. For a proof see \cite{hosseini-2013-sparse-tech}. 
\begin{itemize}

\item[(c)] Consider the left panel of Figure \ref{sin_LB_BDE_curve_CVX_analytic} for the proof. Let $f(x)=\sin(2\pi
x),$ $f_1(x)=f(x)-\sigma$, $f_2(x)=f(x)+\sigma,\;x\in[0,1].$ Define
\[A=(x_A,y_A)=(1/4,f_1(1/4)),\;C=(1/2,f(1/2)=0),\;E=(3/4,f_2(3/4))\]
Also let $B=(x_B,f_1(x_B))$ be a point on the trajectory of $f_1(x),\;x\in[1/4,1/2]$ such that $BC$ is tangent to
$f_1(x)$ trajectory. Then let $D$ be the symmetric image of $B$ with respect to $C$. Then $DB$ is also tangent to
$f_2$ trajectory by symmetry. Then consider the curve (dashed) that goes along $f_1(x)$ trajectory from $A$ to $B$; then goes along the line
segment $BD$; then goes along the $f_2(x)$ trajectory to reach $E$. We claim that this curve has the minimum
possible LB while satisfying the deviation $\sigma$. First note that such a curve satisfies the deviation, $\sigma,$ and can be extended in the same manner to $[0,1]$ (dashed curve). We denote this curve by $g$. Then
note that $Lip(g)$ is the same when applied to the domain $[x_A,x_E]$ or when applied to $[0,1]$. In fact
$Lip(g)$ equals to slope of $BC$  line which we denote by $l$.
 Therefore it only remains to show no other curve achieves this and obtain a strictly smaller LB on $[x_A,x_E]$.
Suppose $h$ is another curve defined on $[x_A,x_E]$ which satisfies the deviation $\sigma$ and has a smaller LB
than $g$ on $[x_A,x_B]$. Without loss of generality (and by symmetry), we can assume that $h(x_C)\leq y_C$ and we
focus on the $[x_B,x_C]$ interval. (If $h(x_C)>y_C$ we repeat the following proof by focusing on $[x_C,x_D]$.)
Then note that $h$ must satisfy $h(x_B) \geq f_1(x_B)=g(x_B)$ since $h$ satisfies the deviation $\sigma$. Now the
line segment from $(x_B,h(x_B))$ to $(x_C,h(x_C))$ will have a slope more than $l$ and this is a contradiction to
$h$ having a smaller LB. To complete the proof it remains to calculate the magnitude of the slope of the $BC$
line segment. This can be found by letting the derivative of $f_1(x_B)$ equal to the slope of $BC$ for $1/4 \leq
x_B \leq 1/2$ and solve that equation for $x_B$:
\[2\pi\cos(2\pi x_B) = \frac{0-(\sin(2\pi x_B)-\sigma)}{1/2-x_B}.\]
Then we  calculate $|2\pi\cos(2\pi x_B)|$ to get the magnitude of the slope.
\end{itemize}
\end{proof}
\begin{figure}
\centering
\includegraphics[width=0.35\textwidth]{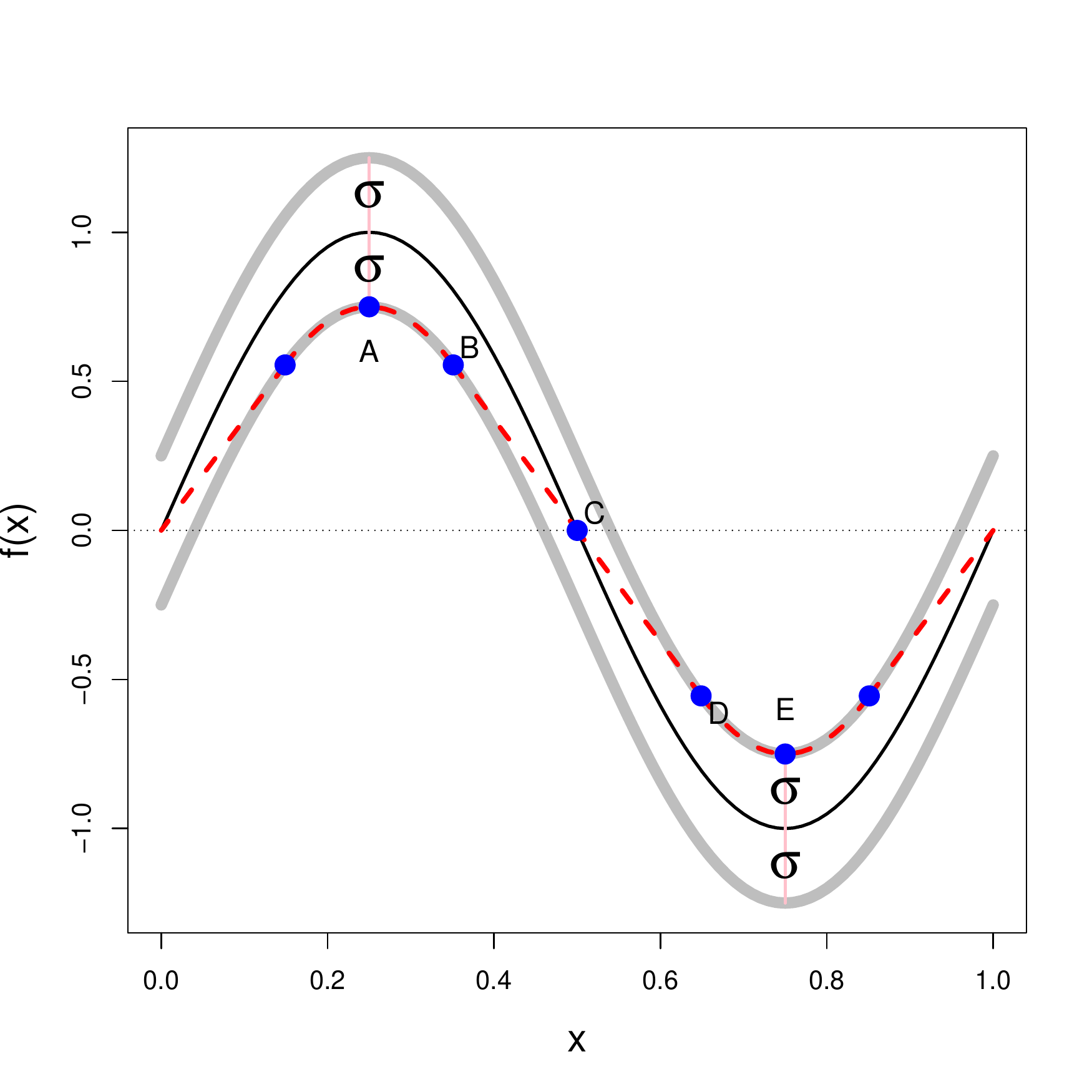}\includegraphics[width=0.35\textwidth]{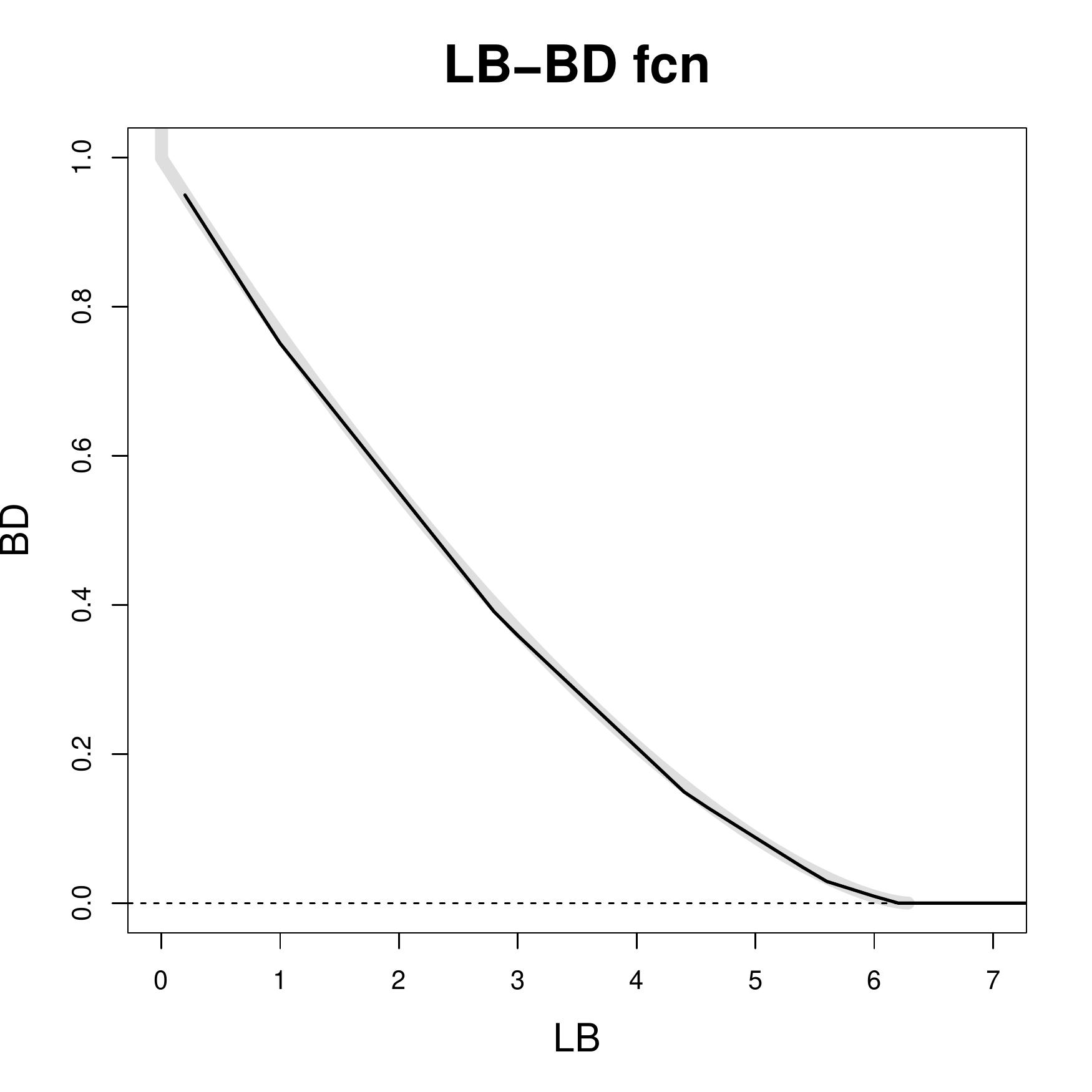}
 \caption{\small (Left Panel): The function $f(x)=sin(2\pi x)$ is give in black with the grey curves mark the boundaries for the curves which deviate from $f$ at most as much
 as $\sigma$. A curve which is inside the boundaries and attains the smallest possible LB is also given.
 (Right Panel): LB-BD curve for $f(x)=\sin(2\pi x),\;x=[0,1]$. Black curve is obtained by analytic solution and
 the grey curve is obtained by solving a convex optimization problem.}
 \label{sin_LB_BDE_curve_CVX_analytic}
\end{figure}

\subsection{Properties of LB-BD function}
This subsection discusses the basic properties of LB-BD function. These properties are useful in providing intuition about the LB-BD curve, as well as calculating
it for given functions.

\begin{lemma}(Elementary Properties of LB-BD function)
Suppose $f:D \rightarrow \R$ is a bounded function and $diam(f)=d$. Then the LB-BD curve of $f$ has the following
properties.
\begin{enumerate}[(a)]
\item $\gamma_f$ and $\gamma_f^{-1}$ are both decreasing functions.
\item $\gamma_f(+\infty)=0,\;\gamma_f(0)=d/2,\;\gamma^{-1}_f(+\infty)=0,\;\gamma^{-1}_f(0)=Lip(f)$.
\item Suppose $f:D \rightarrow \R$ and $f_1:D_1 \subset D \rightarrow \R$ is a restriction of $f$ from domain
$D$ to $D_1\subset D$. Then $\gamma_{f_1}(m) \leq \gamma_f(m),\;m\geq 0$.
\item Suppose $k>0$ and define $f_1(x)=f(kx)$ for $x \in [a/k,b/k]$. Then $\gamma_{f_1}(m) = \gamma_f(m/k).$
\item $\gamma_{kf}(m) = |k| \gamma_f (m/|k|)$.
\end{enumerate}
\label{lemma-elem-LB-BD}
\end{lemma}
\begin{proof} See Appendix.
\end{proof}
\begin{theorem}(Summation Bound on LB-BD function)
Suppose $f=f_1+f_2$.
\begin{enumerate}[(a)]
\item If $m=m_1+m_2$ where $m_1,m_2\geq 0$ then $\gamma_{f}(m) \leq \gamma_{f_1}(m_1) + \gamma_{f_2}(m_2).$
\item  If $\sigma\;=\;\sigma_1\;+\;\sigma_2\;\;$ where $\sigma_1,\sigma_2\geq 0$ then $\gamma^{-1}_{f}(\sigma) \leq \gamma^{-1}_{f_1}(\sigma_1) + \gamma^{-1}_{f_2}(\sigma_2).$
\end{enumerate}
\label{theorem-decomp-LB-BD}
\end{theorem}
\begin{proof} See Appendix.
\end{proof}
\begin{theorem}
Both $\gamma_f$ and $\gamma_f^{-1}$ are convex functions.
\end{theorem}
\begin{proof}
A corollary of the  Theorem \ref{theorem-decomp-LB-BD}.
\end{proof}
\begin{corollary}
Suppose $SPWL(f,g) \leq \sigma$. Then $|\gamma_f(m)-\gamma_g(m)| \leq \sigma/2,\;\forall m\geq 0.$
\end{corollary}
\begin{proof}
Let $h=g-f$ then $SPWL(h,0) \leq \sigma$. Therefore $diam(h) \leq \sigma$ and we conclude
$\gamma_h(0) \leq \sigma/2$.
Now by applying the Decomposition Theorem to $f = g + h$ and for $m_1=m,\;m_2=0$:
\[\gamma_{f}(m) \leq \gamma_{g}(m) + \gamma_{h}(0)   \leq \gamma_{g}(m) + \sigma/2 \;\Rightarrow \; \gamma_{f}(m) - \gamma_{g}(m) \leq \sigma/2.\]
Similarly we can show that: $ \gamma_{g}(m) - \gamma_{f}(m) \leq \sigma/2,$ and thus the proof is complete.
\end{proof}

The following theorem provides a link between the LB-BD of a function and a grid approximation of the function for the 1-dimensional case.
\begin{theorem} (LB-BD Grid Approximation)
Suppose $f: [a,b] \rightarrow \R$ and consider a grid approximation given by
${\bf x}=(x_1,\cdots,x_n)$ and ${\bf y}=(f(x_1),\cdots,f(x_n))$ and denote the grid function by $g$.
Denote the linear interpolation of $g$ on $[a,b]$ by $LI(g)$ and suppose $SPWL(f,LI(g)) \leq \sigma.$
Then $0 \leq \gamma_f(m) - \gamma_g(m) \leq \sigma.$ Note that $\gamma_g$ is calculated with respect to
the domain of $g$ which is ${\bf x}=(x_1,\cdots,x_n)$ and not $[a,b]$.
\label{theorem-LB-BD-grid-approx}
\end{theorem}
\begin{proof} See Appendix.
\end{proof}
For most functions (even simple smooth ones) obtaining the LB-BD curve analytically is not possible. However
using this theorem, we can find a grid for which the grid approximation is arbitrarily close to the original function. Then if we are able to find the LB-BD curve for the gridded function, we can approximate the LB-BD curve of the original function closely. This is also useful from a computational point of view when we are working with data or gridded functions. For
example if we are working with data with ${\bf x}=(x_1,\cdots,x_N)$ and ${\bf y}=(y_1,\cdots,y_N)$ where $N$ is large as we show in the following the LB-BD curve calculation becomes computationally intensive. However we may be able to find sub-grids of ${\bf x}$ and ${\bf y}$: ${\bf x'}=(x_{i_1},\cdots,x_{i_n}); {\bf y'}=(y_{i_1},\cdots,y_{i_n})$, for $1\leq i_1 < \cdots < i_n \leq N$, such that $n << N$ ($n$ is much smaller than $N$) and $SPWL({\bf y},LI({\bf x};{\bf x'},{\bf y'}))\leq \sigma$. We can approximate LB-BD curve of $({\bf x},{\bf y})$ by calculating that of $({\bf x'},{\bf y'})$ and noting that
$\gamma_{({\bf x},{\bf y})} - \gamma_{({\bf x'},{\bf y'})} \leq \sigma$.

\subsection{Calculating LB-BD function}
\label{sect:calculating-LB-BD}
This subsection assumes we have access to gridded data and using that we develop methods to calculate the LB-BD curve. Theorem
\ref{theorem-LB-BD-grid-approx} then can be applied to make a connection to a full curve or a curve defined on a  more fine resolution.
This may seem contradictory to the sparse data situation at first but as we discuss in more details later the LB-BD curve for many applications
does not vary much from one time period to another or we may use the LB-BD curve of a temporal process in one location with dense data for another close location
with sparse data. As an example we show that the LB-BD curve is similar for the temporal process of several central sites for Ozone process in Southern California.

Below we start with a heuristic moving average filtering method for calculating the LB-BD curve for 1-d case. Then we proceed to an exact method
by representing the LB-BD calculation as a  convex optimization problem. This convex optimization method also works for the multidimensional data case. However for the 1-d case, we also present a faster method by representing the problem as a different convex optimization method.

\subsection{Convex optimization method}
This subsection discusses methods for calculating the LB-BD function for gridded functions. Suppose $f:D \subset \R^d$ is a given function for which we like to find the LB-BD curve. To calculate $\gamma_f(m)$, we need to solve:\begin{equation}
\underset{g \in \LB(m)}{\inf} \;\underset{x \in D}{\sup} |f(x)-g(x)|. \label{eqn-deriv-bound}
\end{equation}
Here we present a
method for estimating the LB-BD function by solving Equation \ref{eqn-deriv-bound} using convex optimization,  when $D$ is a finite subset.\\

\noindent{\bf Convex optimization for calculating LB-BD function:}
Suppose $f: D \subset \R^d \rightarrow \R,$ is defined on a finite domain ${\bf x}=(x_1,\cdots,x_n)$ ($D$ is the set defined by the elements of {\bf x}) and takes the values ${\bf y}=f({\bf x})=(y_1,\cdots,y_n)$. Consider an approximation of
${\bf y}=f({\bf x})$ by  $y^{\star}$:
\[y_i={y_i^{\star}}+{r_i},\]
where $r_i$ is the deviation from the true value at $y_i$. This approximation belongs to $\LB(m)$ if
and only if
\[y_i^{\star}-y_j^{\star} \leq m ||x_i-x_j||,\; \forall i,j \in \{1,\cdots,n\},\]
(  \cite{Beliakov-2006}). We conclude that finding the value of $\gamma_f(m)$ is equivalent to
 minimizing $\underset{i=1,\cdots,n}{\max} |r_i|$. Now we pose the convex optimization method:
\begin{enumerate}[(a)]
\item For finding $\gamma_f$:
\begin{align*}
&\mbox{minimize}& \;\; \underset{i=1,\cdots,n}{\max} |r_i|,&\\
&\mbox{subject to}& \;\;r_i-r_j  \leq & m  ||x_i-x_j||  + (y_j-y_i),\;\forall i,j \in \{1,\cdots,n\}.
\end{align*}
\item For finding $\gamma^{-1}_f:$
\begin{align*}
&\mbox{minimize}&\;\; \underset{i=1,\cdots,n}{\max} & |(y_j-y_i+(r_i-r_j))|/||x_i-x_j||,\\
&\mbox{subject to}&\;\;|r_i| \leq & \sigma ,i \in \{1,\cdots,n\},\;\forall i,j \in \{1,\cdots,n\}.
\end{align*}
\end{enumerate}

These problems can then be implemented in the CVX package of Matlab (see \cite{Grant-2008}). (a)
is the minimization of a maximum of absolute values of $n^2$ affine functions and with $n$ affine constraints.
(b) is the minimization of a maximum of absolute values of $n$ affine functions and with $n^2$ affine
constraints.\\

\noindent{\bf Convex optimization for 1-d case:}
Suppose we want to calculate $\gamma_f(m)$ where $f$ is defined on $[a,b]$ and is equal to the linear
interpolation of $a \leq x_1 <x_2 <\cdots<x_n \leq b$ with values $(y_1,\cdots,y_n)$. Then
\[\gamma_f(m)=\underset{g \in \LB(m)}\inf SPWL(f,g)=\gamma_f(m)=\underset{g \in \PL(m)}\inf SPWL(f,g),\]
because $\PL(m)$ is dense in $\LB(m)$. Now suppose a $g\in \PL(m)$ attains $SPWL(f,g)=\sigma$. Because $g$ is
piece-wise linear, $g$ has breakpoints at $a\leq z_1< z_2 < \cdots <z_k \leq b$. Clearly we can assume $z_i$s
include the $x_i$s as we do not require $g$ to change slope at every break point. Moreover  we claim
that there is always a $h\in \PL(m)$ which is as close to $f$ as $g$, $SPWL(f,h) \leq \sigma$,  and only requires
break points at $x_i$s. We define such a $h$ by modifying $g$. We define $h$ to be the linear interpolation of
$x=(x_1,\cdots,x_n)$ with values at $y=(g(x_1),\cdots,g(x_n))$. Then it is clear that $h \in LB(m)$ (because $g$
is) and $SPWL(f,h)=SPWL(f,g)=\sigma$. Any such $h$ can be written as a linear combination of
\[h(x) = c_0 + \sum_{i=1}^n m_i 1_{\{x>x_i\}}(x-x_i) \]
where $1_{x>x_i}=1 \iff x>x_i$. Now using this definition at the breakpoints ${\bf x}=(x_1,\cdots,x_n)$ and the definitions:
\[{\bf y}:=\begin{pmatrix}y_1\\ y_2 \\ y_3 \\ \vdots\\y_n\end{pmatrix},\;
{\bf X}=\begin{pmatrix}1 & 0 & 0 &  \cdots & 0\\
1 & x_2-x_1 & 0 &  \cdots & 0\\
1 & x_3-x_1 & x_3-x_2 &  \cdots & 0\\
\vdots & \vdots & \vdots  & \vdots & \vdots \\
1 & x_n-x_1 & x_n-x_2 &  \cdots & x_n-x_{n-1}\\\end{pmatrix},\;{\bf r}=\begin{pmatrix}r_1\\ r_2 \\ r_3 \\
\vdots\\r_n\end{pmatrix},\;{\bf m}=\begin{pmatrix} m_1 \\ m_2 \\ \vdots \\ m_{n-1} \end{pmatrix},\]
we can write
\begin{eqnarray*}
 {\bf y} &=& {\bf X}\begin{pmatrix} c_0 \\ {\bf m} \end{pmatrix} + {\bf r},\\
\end{eqnarray*}
where $c_0$  is the value of $h$ at $y_1$; $m_1,\cdots,m_{n-1}$ are the slopes at the break points. Then $f$ belongs
to $\ALB(m,\sigma)$ if and only if
\begin{equation}\underset{i=1,\cdots,n-1}{\max}{|m_i|}\leq
m,\;\;\;\underset{i=1,\cdots,n-1}{\max}{|r_i|}\leq \sigma.\label{equation_conditions_ALB_data}
\end{equation}

Additionally if we define ${\bf 1}_n$ to be a column vector of all 1s and of length $n$, for all natural numbers
$n$, we can also write the conditions in \ref{equation_conditions_ALB_data} in the matrix form:
\[ -m {\bf 1}_{n-1} \leq {\bf m} \leq m {\bf 1}_{n-1},\;\;-\sigma {\bf 1}_n \leq {\bf r} \leq \sigma {\bf 1}_n.\]
Or if we use the definition of maximum norm(infinity norm):
$||(x_1,\cdots,x_n)||_{\infty}=\underset{i=1,\cdots,n}{\max}{|x_i|}$, we can write them as $||{\bf m}||_{\infty}\leq m,\;||{\bf r}||_{\infty}\leq \sigma.$
For the periodic case, $\PALB(m,\sigma)$, we need an extra condition which assures that the magnitude of the slope
of the line going from the last point $(x_n,h(x_n))$ to $(b+(x_1-a),h(x_1)=c_0)$ is also less than $m$:
\[-m \leq \frac{\sum_{i=2}^n(x_i-x_{i-1})m_{i-1}}{(b-a)-(x_n-x_1)} \leq m.\]
Now we pose the convex optimization method:
\begin{enumerate}
\item[(a)] For finding $\gamma_f$:
 \begin{eqnarray*}
&\mbox{minimize}\;\;& ||{\bf r}||_{\infty},\\
&\mbox{subject to}\;\;& -m {\bf 1}_{n-1} \leq {\bf m} \leq m {\bf 1}_{n-1}
\end{eqnarray*}
\item[(b)] For finding $\gamma^{-1}_f:$
 \begin{eqnarray*}
&\mbox{minimize}&\;\; ||{\bf m}||_{\infty},\\
&\mbox{subject to}&\;\; -\sigma {\bf 1}_{n-1} \leq {\bf r} \leq \sigma {\bf 1}_{n-1}
\end{eqnarray*}
\end{enumerate}

Table \ref{table-convex-optim-comparison} compares the computation time for 1-d functions using both the fast (1-d) and general method. The computational gain is very significant and grows with the data size.
{\tiny
\begin{table}[H]
 \caption{\small Comparison between the optimization methods (to calculate LB-BD curve) for 1-d and general method. The methods are applied
 to an equally-spaced gridded version of the function $\sin(2\pi x),\,x \in [0,1]$ with various resolutions determined by data size. In each case, the BD value is calculated for LB in $\{0,0.1,0.2,\cdots,10\}$.}
 \centering  \footnotesize
 \begin{tabular}{|l|c|c|c|c|c|c|}
\toprule[1pt]
data size & General method time(s) & Fast method (1-d) time(s) & Time Ratio \\
\midrule[1pt] 
10 & 38 & 19 & 2.0 \\
20 & 58 & 21 & 2.7\\
30 & 49 & 175 & 3.6 \\
40 & 68 & 606 & 8.9 \\
50 & 87 & 2139 & 25 \\
\bottomrule[1pt]
\end{tabular}
\label{table-convex-optim-comparison}
\end{table}}

\subsection{Choosing appropriate parameters using data}
\label{sect:choosing-LB-BD}
In order to be able to use the prediction errors of various methods or for applying the
 $Lipfit$ method, one needs to find appropriate LB and BD. In some applications,
 this can come from the expert knowledge of the practitioner. It is often unreasonable to assume that high-order
derivatives exist and also require the practitioner to know about its magnitude. On the other hand, for many
applications, using physical/chemical/biological properties of the process, the practitioner may obtain a bound
on the rate of change of the process as measured by LB and a small-scale deviation (BD). The small-scale deviation may
refer to the accuracy of the measurement device or small-scale variations of the process. However one does not
need to merely rely on the expert knowledge or the properties of the processes. In the following we show how one
can use available data to get an estimate of these parameters $m$ and $\sigma$.

In the above, we presented a method to calculate LB-BD curve when we have sufficient data from the process under the
study. One question is given an LB-BD curve which pair should be used for fitting the $Lipfit$ method and calculating the prediction errors.
The main method that we discuss here is picking the pair which minimizes the given errors.  \cite{hosseini-2013-sparse-tech} develops  {\it validation} methods
(either using multiple instances of the process or cross-validation). Also after all the goal is to approximate curves when enough data is not available and it may look such
a method is not useful in practice. Here we discuss under what situations this method may be useful by giving concrete situations where the methods can be applied.
Also in Section \ref{sect:application}, we apply the methods to air pollution data.

\subsection{Prediction errors given the LB-BD curve}
For a function $f:D \subset \R^d \rightarrow \R$ observed at given points ${\bf x}=(x_1,\cdots,x_n)$ and with values equal to
${\bf y}=(y_1,\cdots,y_n)$, we found the errors for estimating it over the entire domain $D$: $IE, DIE, SPWE,
DSPWE, DSPWE$ for each of the methods (e.g. $LI$ and $Lipfit$) and for a given fixed pair of LB and BD. Now
suppose instead of one single pair, a {\it partial} LB-BD curve
\[\gamma_f:U\subset \R^{\geq 0} \rightarrow \R^{\geq 0},\]
is given. We can think of $U$ as a subset of $\R^{\geq 0}$, where information is available about $f$. Then we
can extend the above errors of estimating $f$ on $D$ by taking the infimum over all the available pairs of
LB-BD. Suppose $approx$ denote the method (for example $approx=Lipfit$) and $E$ the error measure (for example
$E=DSPWE$), then the minimal error of estimating $f$ given $\gamma_f$ is defined as follows:
\[\Upsilon\{E,approx,f,D,{\bf x},{\bf y}\;|\;\gamma_f\}=\inf_{m\in U} E\{approx,f,D,{\bf x},{\bf y}\;|\;m,\gamma_f (m)\}.\]
Since $f$ belongs to all $(m,\gamma_f (m))$, when the infimum is obtained by some $m_0 \in U$, we can apply the
method $approx$ with that $(m_0,\gamma_f(m_0))$ to get the error $E=\Upsilon$, therefore minimizing the error
on $D$ as much as possible. If the infimum is not obtained for any small $\epsilon$ there is $m_0 \in U$ so
that $E$ is within a radius of $\epsilon$ of $\Upsilon$.

In applications, for a given function there
are infinitely many pairs of LB-BD, $(m,\sigma)$, for which the given function belongs to $\ALB(D,m,\sigma)$. Then we can use the above idea to define a method for picking an appropriate $(m,\sigma)$ for fitting and calculating the errors. We introduce the Prediction Error Minimization Method (PEM)  as follows. For a given 
LB-BD curve and data set  $({\bf x},{\bf y}),$ pick the pair which minimizes the appropriate
prediction error of interest:
\[PEM:\;({\bf x},{\bf y},\gamma_f) \mapsto (m,\gamma_f(m)),\]
or more formally attains $\Upsilon\{E,approx,f,D,x,y\;|\;\gamma_f\}$. In practice it may not be obvious if any (unique) pair attain the infimum, however one can always consider a grid for $m$ and pick (one of the) the grid point for which $E\{approx,f,D,x,y\;|\;m,\gamma_f (m)\}$ is minimized. We will use and expand on this method in
Section \ref{sect:application}. Validation methods such as cross validation which are widely used 
in picking the parameters of statistical models (e.g. in \cite{book-hastie-tib}) can be used to pick $(m,\sigma)$ as discussed in \cite{hosseini-2013-sparse-tech} .

\section{Application to air pollution data}
\label{sect:application}
This section applies the methods developed in this work to air pollution data.
We are interested in approximating the biweekly averaged air pollution (Ozone) process in homes and schools in
Southern California during 2005, using three biweekly measurements in the spring, summer and winter. The moving
average process refers to a process for which the value of the process on each day is the average of the process
in 15 days centered around that day. (For data collection in the study, measurement filters are placed in the
school for two-week periods to collect aggregated air pollution levels.)  We also have access to 11 central
sites for which complete data are available during 2004--2007. To be more concrete denote the biweekly
averaged pollution process by $Y(s,t)$ at the location $s$ for which three times during 2005 are available:
$Y(s,t_1),Y(s,t_2),Y(s,t_3)$. We denote the 11 central site locations  by $s_1,s_2,\cdots,s_{11}$.

Figure \ref{LB_BD_comm_nonper.pdf} depicts the calculated LB-BD functions for the 11 communities  for both
$\ALB$ and $\PALB$ families. There is good consistency among the
curves across the central sites, except for Santa Maria for which the curve is visibly placed below all the
other curves for both families. This is because Santa Maria is a much cleaner community with lower levels of
Ozone and its variation across the year. This figure suggests that if we use the LB-BD curve from one location
for a location which is not too far or too different from the location of interest then the results will be
reliable.

\begin{figure}
\centering
\includegraphics[width=0.4\textwidth]{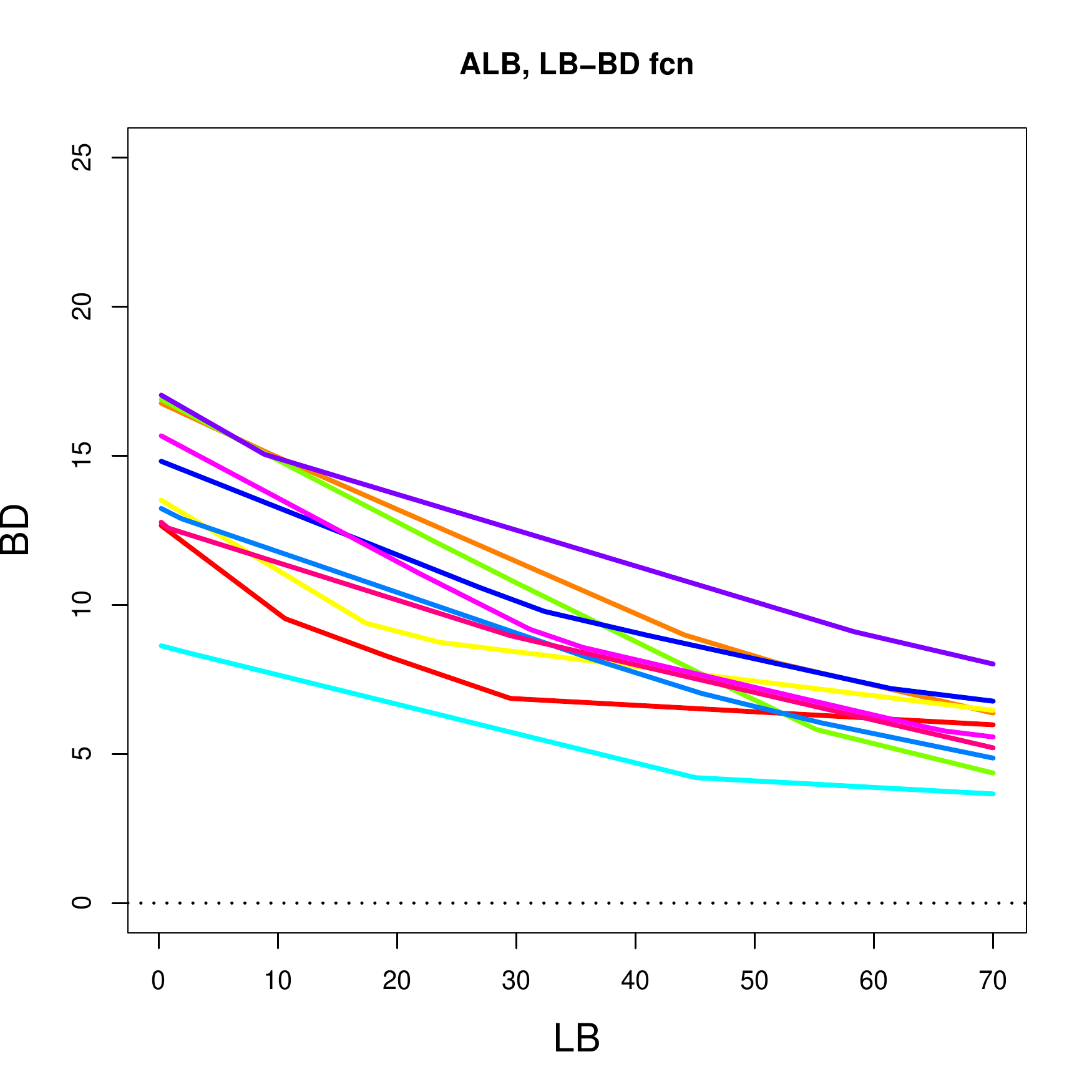}\includegraphics[width=0.4\textwidth]{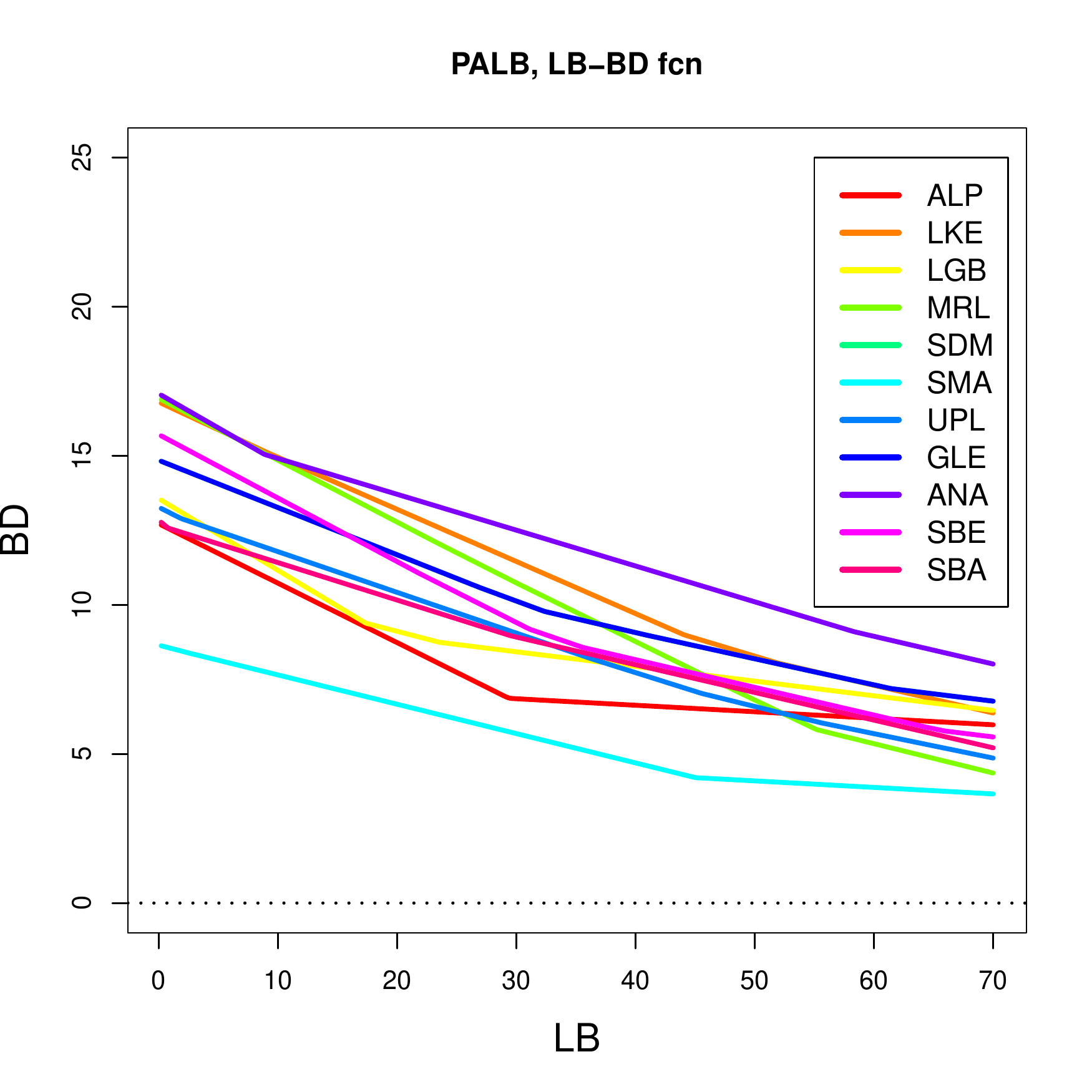}
\caption{\small LB-BD functions for the temporal biweekly Ozone process for 11 central stations in Southern California.} \label{LB_BD_comm_nonper.pdf}
\end{figure}
Here we describe two scenarios one may be interested in for prediction:  (1) We assume that the new location
$s$ does not belong to any of the communities and therefore does not have a very similar weather pattern to any central site with complete data. (2) We
assume that the station belongs to one of the communities and therefore there is a rather close central site but
there is no guarantee that the seasonal patterns completely match. Also assume in (1) there are no observations in
nearby locations with similar weather patterns to borrow strength across space to build a complex statistical
model; in (2) there is a nearby station with complete data across the year, but there is no guarantee the
seasonal patterns of the air pollution process in the two locations are exactly the same. In situation (1), as discussed
above, we can choose nearby locations with complete data and then apply the methods described above to find
appropriate LB-BD for those locations in order to use for the location with incomplete data. Note that all
we need is a bound on $m$ and $\sigma$ and therefore we can use slightly larger $m$, $\sigma$ from what we have
found, if there is more doubt about the similarity of the other locations in terms LB-BD.

In situation (2), suppose the closest central site location is $s_i$. Then we can calculate the difference process:
$D(t_j)=Y(s,t_j)-Y(s_i,t_j)$. If the seasonal patterns are the same, we must have $D(t_1)\approx D(t_2) \approx
D(t_3)$. Note that this does not guarantee that the seasonal patterns are the same and we are obliged to rely
also on some expert knowledge. However if such a knowledge is available one may suggest to take the average of
the differences $D=\sum_{j=1}^3 D(t_j)/3$ and then return $Y(s_i,t)-D$ as an approximation of $Y(s,t)$ for all $t$ during 2005.
Instead of averaging $D(t_j)$'s: we approximate the difference process $D(t)$ using $Lipfit$ to get a complete series $\hat{D}(t)$ and then return $Y(s_i,t)-\hat{D}(t)$ as our approximation of
$Y(s,t)$. Thus we accommodate the possibility that the difference between the two processes
$Y(s,t)$ and $Y(s_i,t)$ varies over time and $Lipfit$ tries to capture that variation.

To give example for scenarios (1) and (2), we choose three central sites in Upland (UPL), Long Beach (LGB) and Anaheim
(ANA). We focus on approximating the curves for Upland and Long Beach. For scenario (1) we calculate the LB-BD curve for Upland, Long Beach and for 
scenario (2) we calculate the LB-BD curve for the difference of the two processes with Anaheim. Then for each
scenario and for both $\ALB$ and $\PALB$, we use the Prediction Error Minimization method (PEM) to pick a LB-BD pair.
Table \ref{table-airpol-LB-BD-curve-errors} presents the results of calculating the minimal error, $\Upsilon$, to
the Ozone process during 2005 and in the locations Long Beach and Upland and the difference processes
of these two locations with Anaheim. For each case the LB-BD curve is calculated using the data and then the
minimal error for various integral and point-wise error measures when three data points are available are reported in the table. 
The pair (LB, BD) for which the error is minimized is also reported. We have
done this for both $\ALB$ and $\PALB$ families. The results observed in the table are as follows: $\PALB$ method
generally works better for these data due to the approximate periodicity of Ozone process; the data-informed
errors are considerably smaller than their non-informed versions  in some
cases, showing that developing and using the data-informed errors is worthwhile; the $Lipfit$ method has
outperformed $LI$ in some cases and is never inferior to $LI$.

{\tiny
\begin{table}[H]
 \caption{\small The errors $(IE, DIE);(SPWE, DSPWE[LI], DSPWE[Lipfit])$ with three data points.}
 
  \centering  \footnotesize
\begin{tabular}{|l|cc|cc|}

\toprule[1pt]
 &  $\ALB$    & &  $\PALB$ &   \\
\midrule[1pt] 
Process &  errors   & (LB, BD)  & errors & (LB, BD) \\ 
\midrule[1pt] 
LGB &  (12, 10); (15, 15, 15)   &  (36, 9) &(13, 8.1); (15, 15, 13)  & (29, 11)  \\
UPL &   (13, 9.9); (15, 15, 15)  &  (17, 12) &(13, 7.1); (14, 14, 14)  &  (17,12)  \\
ANA-LGB &   (6.6, 5.2); (8.3, 8.3, 8.3)  &(20, 5) & (7.1, 4.2); (8.4, 7.1, 6.5)   & (16, 5.7) \\
ANA-UPL &  (6.4, 4.6); (6.9, 6.9, 6.9)  & (5.5, 6.0) &(6.6, 3.9); (7, 6.1, 5.1)  &  (5.4, 6.1) \\
\bottomrule[1pt]
\end{tabular}
\label{table-airpol-LB-BD-curve-errors}
\end{table}}

\section{Discussion and future directions}
\label{sect:discussion}

This work developed a framework for fitting functions with sparse data. At first we considered a framework based on measuring the variation of the functions
by Lipschitz Bound, also considered by \cite{Sukharev-1978} and \cite{Beliakov-2006}. The limitation in using such a framework is due to the fact that many processes in practice, despite revealing a slow global variation, have some smaller scale variations
which cause the Lipschitz Bound to be too large to be useful in fitting or calculating the prediction errors.
Thus this work extended this framework by accepting a number $m$ as Lipschitz Bound up to a Bound Deviation,  $\sigma,$ if the function of interest can be approximated by another which accepts $m$ as LB and does not deviate from the original function more than $\sigma$  (in terms of sup norm). Using this framework, we found reasonable fits and prediction errors for functions which do not admit a small enough Lipschitz Bound.

 Another key idea we introduced is the formalization of the trade-off between the variation measure (LB) and the deviation measure (BD as measured by sup norm here) which is summarized in a non-increasing convex curve -- LB-BD function.
 We provided convex optimization methods to calculate the LB-BD curve using data or gridded versions of the functions under study and provided the connection of the LB-BD curve of a gridded function to its more fine-resolution version. Given the LB-BD curve for a function, we develop a method to find an appropriate LB-BD pair to apply $Lipfit$ -- a pair which depended on the data. Given the LB-BD curve, we also calculated the minimal prediction error, e.g.\;$DSPWE$, by minimizing it across the LB-BD curve. In the background section, we made some connection between this work and some smoothing methods such as smoothing splines. In fact the smoothing spline method merely used a different variation measure $\int_D ||f''(x)||dx$ and deviation measure: the sum of the square of the difference between the observed data and fitted. This immediately leads to the idea of generalizing this framework by choosing various variation and deviation measures and define a variation-deviation curve
 (an extension of LB-BD curve) for each case.

 Some other important extensions and open problems are: (1) Given a curve decide if the curve can be LB-BD curve for a function; such a curve must be non-increasing and convex as we show in this work, but is that enough? (2) If we consider a random process $\{Y(t)\}, t \in T$ for some space $T$ (time, space, spatial-temporal), for each instance of this process we can calculate an LB-BD curve. Thus LB-BD curve is a random quantity. In this work we assumed that this random quantity does not vary much from one instance to another; or it is applicable from one time interval to another; or if we use the LB-BD curve of a comparable data set (for example a nearby station) the results are reliable. In fact using some simulations and real air pollution data we showed that this can be the case. However, it is interesting to investigate the LB-BD variability for random processes and it may be even useful to develop parametric and non-parametric models for LB-BD curve as a random quantity. We also leave these problems for future research. (3) In this we work, we developed the Prediction Error Minimization (PEM) to pick a LB-BD curve to minimize the error of interest for example $DSPWE$ over the domain of interest, given data. Then we used that pair for applying $Lipfit$ and calculating the prediction errors. The restriction of this method lies in the fact that we used the same LB-BD pair to  approximate all points. Alternatively we can allow picking a different LB-BD pair for each given data point $x$, a method which in general can improve the $pef$ at any given point at a cost of more computations.\\

\textbf{Acknowledgements:}
This work was partially supported by National Institute of Environmental Health Sciences (5P30ES007048, 5P01ES011627, 5P01ES009581); United States Environmental Protection Agency (R826708-01, RD831861-01); National Heart Lung and Blood Institute (5R01HL061768, 5R01HL076647); California Air Resources Board contract (94-331); and the Hastings Foundation. The first author was also supported by research grants from the Japanese Society for Promotion of Science (JSPS). The authors gratefully acknowledge useful discussions with Drs Duncan Thomas, Jim Gauderman, and Meredith Franklin.

\section{Appendix: Proofs}

%%%%% LB-BD curve properties

%% elementary properties
\begin{proof} Lemma \ref{lemma-elem-LB-BD}.
\begin{enumerate}[(a)]
\item This is straightforward from the properties of infimum.
\item For $m=+\infty,$ $f \in \LB(m)$ and therefore $\gamma_f(m)=0$. Also only constant functions satisfy $m=0$
and therefore $\gamma_f(m)=d/2$ as the constant function $g(x)=(\underset{z \in [a,b]}{\sup} f(z) - \underset{z
\in [a,b]}{\inf} f(z))/2$ minimizes $SPWL(f,g)$. For $\sigma=+\infty$ any bounded function, $g$, satisfies $SPWL(f,g)\leq \sigma$ including any constant
function $g=c$ for which we have $Lip(g)=0$. The only function, $g$, which satisfies $SPWL(f,g)=0$ is $f$ and
therefore $\gamma^{-1}(0)=Lip(f)$.
\item Obvious from the definition.
\item Suppose $\gamma_f(m)=\sigma$ which means $\sigma=\underset{g \in \LB(m)}{\inf} SPWL(f,g).$
Now let us calculate the quantity of interest $\gamma_{f_1}(m)$:
\begin{align*}
\gamma_{f_1}(m)=&\underset{g_1 \in \LB(m)}{\inf}SPWL(f_1,g_1)\\
=&\underset{g_1 \in \LB(m)}{\inf}\;\; \underset{x \in [a/k,b/k]}{\sup}|f_1(x)-g_1(x)|\\
=&\underset{g_1 \in \LB(m)}{\inf}\;\; \underset{x \in [a/k,b/k]}{\sup}|f(kx)-g_1(kx/k)|\\
=&\underset{g_1 \in \LB(m)}{\inf}\;\; \underset{y \in [a,b]}{\sup}|f(y)-g_1(y/k)|\\
=&\underset{g_2(y)=g_1(y/k);\,g_1 \in \LB(m)}{\inf}\;\;\;\; \underset{y \in [a,b]}{\sup}|f(y)-g_2(y)|\\
=&\underset{g_2 \in \LB(m/k)}{\inf}\;\; \underset{y \in [a,b]}{\sup}|f(y)-g_2(y)|\\
=&\gamma_f(m/k).
\end{align*}
\item Define $f_1(x)=kf(x)$ on the same domain. Then we have
\begin{align*}
\gamma_{f_1}(m)=&\underset{g_1 \in \LB(m)}{\inf}\;\;\underset{x \in [a,b]}{\sup}|f_1(x)-g_1(x)|\\
=&\underset{g_1 \in \LB(m)}{\inf}\;\;\underset{x \in [a,b]}{\sup}|kf(x)-g_1(x)|\\
=&\underset{g_2=g_1/k;\,g_1 \in \LB(m)}{\inf}\;\;\;\;\underset{x \in [a,b]}{\sup}|kf(x)-kg_2(x)|\\
=&\underset{g_2 \in \LB(m/k)}{\inf}\;\;\;\;\underset{x \in [a,b]}{\sup}|k||f(x)-g_2(x)|\\
=&|k|\gamma_f(m/k).
\end{align*}
\end{enumerate}
\end{proof}
\begin{proof} Theorem \ref{theorem-decomp-LB-BD}.
\begin{enumerate}[(a)]
\item Suppose $\gamma_{f_i}(m_i)=\sigma_i,\;i=1,2$. Then for any (small) $\epsilon>0$, there exist functions
$g_i \in \LB(m_i)$ such that $SPWL(f_i,g_i) \leq \sigma_i - \epsilon,\;i=1,2$. Then clearly $g=g_1+g_2 \in
\LB(m)$ and we have
\begin{align*}
\gamma_f(m) \leq & SPWL(f,g) \leq SPWL(f_1,g_1) + SPWL(f_2,g_2)\\
 \leq & \sigma_1 + \sigma_2 - 2\epsilon = \gamma_{f_1}(m_1) + \gamma_{f_2}(m_2) - 2\epsilon.
\end{align*}
Since above holds for any $\epsilon>0$, we conclude $\gamma_f(m) \leq  \gamma_{f_1}(m_1) + \gamma_{f_2}(m_2).$

\item Suppose $\gamma^{-1}_{f_i}(\sigma_i)=m_i,\;i=1,2$ and fix bounded $f_1,f_2$ so that $f=f_1+f_2$ and let
$d_i=diam(f_1),i=1,2,\;d=\max\{d_1,d_2\}$. Then for any (small) $\epsilon>0$, there exist functions $g_i \in
\LB(m_i+\epsilon)$ such that $SPWL(f_i,g_i) \leq \sigma_i,\;i=1,2$. Clearly $g=g_1+g_2 \in \LB(m+2\epsilon)$ and
define
\[c=\frac{m}{m+2\epsilon},\;\tilde{g}=cg.\]
 Then we have $\tilde{g}\in \LB(m)$ and
\begin{eqnarray*}
\gamma^{-1}_f(m) \leq SPWL(f,\tilde{g}) \leq SPWL(f_1,cg_1) + SPWL(f_2,cg_2)\\
 \leq SPWL(f_1,g_1) +  SPWL(g_1,cg_1) + SPWL(f_2,cg_2) + SPWL(g_2,cg_2)\\
 \leq \sigma_1 + \sigma_2 + (1-c)diam(g_1) + (1-c) diam(g_2)\\
  \leq \sigma_1 + \sigma_2 + (1-c)(d+\sigma_1)+(1-c)(d+\sigma_2)\\
  \leq  \sigma_1 + \sigma_2 + (1-c)(d+\sigma_1+\sigma_2)\\
= \gamma^{-1}_{f_1}(m_1) + \gamma^{-1}_{f_2}(m_2) + (1-c)(d+\sigma_1+\sigma_2)
\end{eqnarray*}
Since the above holds for any $\epsilon>0$, $(1-c)=2\epsilon/(m+2\epsilon)$ can become arbitrarily small. Now since
$(d+\sigma_1+\sigma_2)$ is fixed, we can omit the last term and conclude $\gamma^{-1}_f(m) \leq
\gamma^{-1}_{f_1}(m_1) + \gamma^{-1}_{f_2}(m_2)$.

\end{enumerate}
\end{proof}

\begin{proof}
Theorem \ref{theorem-LB-BD-grid-approx}.\\
First note that, clearly $\gamma_f(m) - \gamma_g(m) \geq 0$ as $f$ is defined on a domain which includes the domain of $g$.
Now suppose $ \gamma_g(m) = \sigma_1$. Then we claim that $\gamma_{LI(g)}(m)=\sigma_1$ also.
$\gamma_{LI(g)}(m) \geq \sigma_1$ is obvious because the domain of $LI(g)$  includes that of $g$.
To show that $\gamma_{LI(g)}(m) \leq \sigma_1$, for any $\epsilon>0$ we will show  $\gamma_{LI(g)}(m) \leq \sigma_1+\epsilon.$
Since  $ \gamma_g(m) = \sigma_1$, there is a grid function $h$, defined on ${\bf x}$, such that
$SPWL(g,h) \leq \sigma_1 +\epsilon$ and $h \in \LB(m,{\bf x})$. Now consider the linear interpolation of $h$ on the interval
$[a,b]$ and denote it by $LI(h)$. Then we  have $LI(h) \in \LB(m,[a,b])$. But we also have $SPWL(LI(g),LI(h)) \leq \sigma_1 + \epsilon$ because the supremum distance is obtained at the break points for piece-wise linear functions.
Therefore \[SPWL(f,LI(h)) \leq SPWL(f,LI(g)) + SPWL(LI(g),LI(h)) \leq \sigma + \sigma_1 + \epsilon,\;\forall \epsilon \geq 0.\]
We conclude $\gamma_f(m) \leq  \sigma + \sigma_1 = \sigma +  \gamma_g(m),$ which completes the proof.
\end{proof}

\end{document}